\documentclass[english,15pt]{article}
\usepackage[T1]{fontenc}
\usepackage[latin1]{inputenc}
\usepackage{geometry}
\usepackage{float}
\usepackage{mathrsfs}
\usepackage{amsmath}
\usepackage{amssymb}
\usepackage{graphicx}
\usepackage{hyperref}
\usepackage[color=yellow]{todonotes}
\hypersetup{
    colorlinks=true,
    linkcolor=blue,
    filecolor=magenta,      
    urlcolor=cyan,
    citecolor = blue,
}

\usepackage{yhmath}
\makeatletter

\usepackage{bm}
\numberwithin{equation}{section}

\usepackage{algorithm}
 \usepackage{algorithmicx}
\floatstyle{ruled}
\newfloat{algorithm}{tbp}{loa}
\providecommand{\algorithmname}{Algorithm}
\floatname{algorithm}{\protect\algorithmname}

\usepackage{algorithm}
 \usepackage{algorithmicx}

\usepackage{amsthm}

\usepackage{mathrsfs}
\usepackage{amssymb}

\usepackage{amsfonts}

\usepackage{epsfig}

\usepackage{bm}

\usepackage{mathrsfs}

\usepackage{enumerate}

\@ifundefined{definecolor}{\@ifundefined{definecolor}
 {\@ifundefined{definecolor}
 {\usepackage{color}}{}
}{}
}{}

\usepackage{subfig}\usepackage[all]{xy}

\newtheorem{theorem}{Theorem}[section]
\newtheorem{lem}{Lemma}[section]
\newtheorem{rem}{Remark}[section]

\newcounter{hypA}
\newenvironment{hypA}{\refstepcounter{hypA}\begin{itemize}
  \item[({\bf A\arabic{hypA}})]}{\end{itemize}}
\newcounter{hypB}

\newcounter{hypD}

\usepackage{babel}\date{}

\makeatother

\begin{document}

\begin{center}

{\Large \textbf{Unbiased Parameter Estimation of Partially Observed Diffusions using Diffusion Bridges}}

\vspace{0.5cm}

MIGUEL ALVAREZ$^{1}$ \& AJAY JASRA$^{2}$ 

{\footnotesize $^{1}$Applied Mathematics and Computational Science Program, \\ Computer, Electrical and Mathematical Sciences and Engineering Division, \\ King Abdullah University of Science and Technology, Thuwal, 23955-6900, KSA.} \\
{\footnotesize $^{2}$School of Data Science,  The Chinese University of Hong Kong, Shenzhen,  Shenzhen,  CN.}\\
{\footnotesize E-Mail:\,} \texttt{\emph{\footnotesize miguelangel.alvarezballesteros@kaust.edu.sa}}, \\
\texttt{\emph{\footnotesize ajayjasra@cuhk.edu.cn}}

\begin{abstract}
In this article we consider the estimation of static parameters for partially observed diffusion processes with discrete-time observations over a fixed time interval.   In particular,  when one only has access to time-discretized solutions of the diffusions we build upon the works of \cite{ub_par,ub_grad} to devise a method that can estimate the parameters without time-discretization bias.  We leverage an identity associated to the gradient of the log-likelihood associated to diffusion bridges,  which has not been used before.
Contrary to the afore mentioned methods,   the diffusion coefficient can depend on the parameters and our approach facilitates the use of more efficient Markov chain sampling algorithms.  We prove that our estimator is unbiased with finite variance and demonstrate the efficacy of our methodology in several examples.
\end{abstract}

\end{center}

\section{Introduction}

Partially observed diffusion processes associated to discrete time observation problems can be found in a wide variety of applications such as finance,  econometrics and engineering;  see for example \cite{ub_par,cappe,chada,ub_grad} and the references therein.  The problem that we consider is as follows.  We have access to discrete time data,  which at each observation time are assumed to depend on the position of an unobserved (latent) diffusion process.  Added to this is a collection of unknown parameters which relate to both the latent diffusion and the observation density.  On the basis of a fixed data set,  we seek to estimate the unknown parameters of the statistical model. This problem  has been considered by many articles and a non-exhaustive list is \cite{ub_par,beskos1,chada,ub_grad,bayes_mlmc}. In most of these papers,  the authors consider a time-discretization of the diffusion process,  such as the Euler-Maruyama method and then use advanced computational methods based on Markov chain Monte Carlo (MCMC) and/or in combination with stochastic approximation (SA) methods, to estimate the parameters.  In many cases there is a time-discretization error in parameter estimates,  although \cite{ub_par,beskos1,chada,ub_grad} all show how to remove this time discretization error in some form.

We consider estimating the parameters of the model using likelihood-based methods  focussing on using the gradient of the log-likelihood (or score function).   In connection to this approach,  many approaches often resort to using the Girsanov change of measure approach (e.g.~\cite{ub_par,ub_grad}).  One of the main issues is then the computation of the gradient,  which using the Girsanov formula in \cite{ub_par,ub_grad} does not allow one to have parameter dependence in the diffusion coefficient,  except when the diffusion process is in one dimension and the Lamperti transformation is available.  In \cite{ub_grad} (see also \cite{beskos,non_synch,stanton}) the authors detail how one might use the ideas of diffusion bridges (see e.g.~\cite{schauer,vd_meulen_guided_mcmc}) to introduce a new identity for the score function that does not require that the diffusion coefficient is independent of the parameter value.  The objective of this work is to leverage upon this new identity to create a parameter estimation method for partially observed diffusions that 
is unbiased,  in a sense that the expectation of the estimator is equal to the true maximizer of the log-likelihood where the exact model (no time discretization) has been fitted. 

The line of methodology that we follow is based upon the double randomization method of \cite{ub_par} (based on the ideas of \cite{rhee},  see also \cite{matti}),  which has been empirically seen to be computationally more efficient than \cite{ub_grad}.  The approach of 
\cite{ub_par} is based upon two critical methods:  MCMC and SA.  In particular,  the idea is to use an MCMC method to sample from a particular probability distribution associated to the time-discretized diffusion (the smoother) and couplings there-of (full explanations are given in Sections \ref{sec:problem} and \ref{sec:method}) and to insert this into a SA based approach; this is simply randomized Markovian SA (MSA),  see \cite{andr} for example.  The contributions of this article then can be summarized as follows:
\begin{itemize}
\item{We develop a new MCMC method for the smoother in continuous time.}
\item{We develop a new MCMC method for the smoother when time is discretized and show how this algorithm can be coupled.}
\item{We show how these MCMC methods can be used to obtain provably unbiased and finite variance estimators of the unknown parameters.}
\item{We implement our methodology for several examples.}
\end{itemize}
The MCMC approaches that were used in \cite{ub_par} are based on the conditional particle filter (CPF) associated to \cite{mlpf} and coupled conditional particle filter (CCPF); see \cite{andrieu,ub_grad} and the references therein.
The bridge representation detailed in \cite{ub_grad} allows one to define a CPF in continuous time and employing backward sampling \cite{lind,whiteley} which is well known to perform better than the CPF \cite{singh}.  This continuous time representation is important as it facilitates a simple time-discretized version,  which as the time-discretization becomes finer converges to the original continuous time CPF with backward sampling.  In the approaches of \cite{ub_par,ub_grad} such a backward sampling CPF is possible,  but the variance of the weights employed in such an algorithm will explode exponentially as the 
time-discretization becomes finer; the algorithm becomes effectively useless as the time-discretization becomes finer.  This point has been realized in \cite{non_synch,stanton},  but to the best of our knowledge has never been implemented in practice.  The time discretized CPF is then extended to a coupled CCPF and these ingredients suffice to adopt the framework of \cite{ub_par}.  Under assumptions we then prove that our estimator is unbiased and of finite variance and then implement our methodology in several examples.  We remark that \cite{ub_par} did not prove that their estimators are of finite variance.

This article is structured as follows.  In Section \ref{sec:problem} the exact problem of interest is detailed.
Section \ref{sec:method} details the methodology that we use to estimate the parameters unbiasedly.  
Section \ref{sec:theory} presents our theoretical results and their implications.
Section \ref{sec:numerics} provides numerical simulations to demonstrate the methodology that we develop.
Appendix \ref{app:math} gives our mathematical proofs relating to the result in Section \ref{sec:theory}.

\section{Problem}\label{sec:problem}

\subsection{Model}

As the problem we consider is very similar to that in \cite{ub_par,ub_grad},  there is naturally some overlap in the presentation here and those afore-mentioned articles.
We consider the following diffusion process on the filtered probability space $(\Omega,\mathscr{F},\{\mathscr{F}_t\}_{t\geq 0},\mathbb{P}_{\theta})$, $\theta\in\Theta\subseteq\mathbb{R}^{d_{\theta}}$ with $\Theta$ being an open set.
\begin{equation}
dX_t = \mu_{\theta}(X_t)dt + \sigma_{\theta}(X_t)dW_t,\quad\quad X_0\sim\nu_{\theta},\label{eq:diff}
\end{equation}
where $X_t\in\mathbb{R}^d$, $\nu_{\theta}$ a probability measure on $(\mathbb{R}^d,\mathcal{B}(\mathbb{R}^d))$, $\mathcal{B}(\mathbb{R}^d$ the Borel sets on $\mathbb{R}^d$,  
with Lebesgue density denoted $\nu_{\theta}$ also,  for each $\theta\in\Theta$
$\mu_{\theta}:\mathbb{R}^d \rightarrow\mathbb{R}^d$, 
and $\sigma_{\theta}:\mathbb{R}^d \rightarrow\mathbb{R}^{d\times d}$
and $\{W_t\}_{t\geq 0}$ is a standard $d-$dimensional Brownian motion.
Denote $\mu_{\theta}^j$ as the $j^{th}-$component of $\mu_{\theta}$, $j\in\{1,\dots,d\}$
and $\sigma_{\theta}^{j,k}$ as the $(j,k)^{th}-$component of $\sigma_{\theta}$, $(j,k)\in\{1,\dots,d\}^2$.
As already stated,  the diffusion model is different to that in \cite{ub_par,ub_grad},  in that it allows dependence of
the diffusion coefficient on the unknown $\theta$.
We assume the following, referred to as (D1) from herein:
\begin{quote} 
The coefficients, for any fixed $\theta\in\Theta$, $\mu_{\theta}^j\in \mathcal{C}^2(\mathbb{R}^d)$ and $\sigma_{\theta}^{j,k} \in \mathcal{C}^2(\mathbb{R}^d)$, for $(j,k)\in\{1,\ldots, d\}^2$. 
Also for any fixed $x\in\mathbb{R}^d$, $\mu_{\theta}^j,\sigma_{\theta}^{j,k}(x)\in\mathcal{C}^1(\Theta)$ , $(j,k)\in\{1,\dots,d\}^2$.
In addition, $\mu_{\theta}$ and $\sigma_{\theta}$ satisfy 
\begin{itemize}
\item[(i)] {\bf uniform ellipticity}: $\Sigma_{\theta}(x):=\sigma_{\theta}(x)\sigma_{\theta}(x)^{\top}$ is uniformly positive definite 
over $x\in \mathbb{R}^d$;
\item[(ii)] {\bf globally Lipschitz}:
There exists a constant $C<\infty$ such that for every $\theta\in\Theta$ we have
$|\mu_{\theta}^j(x)-\mu_{\theta}^j(x')|+|\sigma_{\theta}^{j,k}(x)-\sigma_{\theta}^{j,k}(x')| \leq C \|x-x'\|$ 
for all $(x,x') \in \mathbb{R}^d\times\mathbb{R}^d$, $(j,k)\in\{1,\dots,d\}^2$. 
\end{itemize}
\end{quote}
The assumption (D1) is not referred to again.  We remark that the uniform ellipticity condition could be relaxed, by considering the methodology of \cite{bierkens} but this is a topic for future work.

We will now consider a sequence of random variables $(Y_{1},\dots,Y_{T})$, where $Y_{p}\in\mathsf{Y}$, that are assumed to have joint conditional density ($T\in\mathbb{N}$)
$$
p_{\theta}(y_{1},\dots,y_{T}|\{x_s\}_{0\leq s\leq T}) = \prod_{k=1}^T g_{\theta}(y_k|x_{k})
$$
where,  for every $\theta\in\Theta$, $g_{\theta}:\mathbb{R}^d\times\mathsf{Y}\rightarrow\mathbb{R}^+$, and for any $(\theta,x)\in\Theta\times\mathbb{R}^d$, $\int_{\mathsf{Y}}g_\theta(y|x)dy=1$ and $dy$ is the dominating measure.  $(Y_1,\cdots,Y_T)$ the observations at unit times $\{1,\cdots,T\}$ where the joint conditional density $p_{\theta}$ is a function of the values of $(Y_1,...,Y_T)$ and $(X_1,...,X_T)$. The general observations at a finite number of arbitrary times $(Y_{t_1},...,Y_{t_k})$ with $t_1<...<t_k$,  can easily be considered,  but we keep to the unit observation times for notational simplicity only.
If one considers realizations $(y_1,...,y_T)\in\mathsf{Y}^T$ of the random variables $(Y_{1},\dots,Y_{T})$, then we have the state-space model with marginal likelihood:
$$
p_{\theta}(y_{1},\dots,y_{T}) = \mathbb{E}_{\theta}\left[\prod_{k=1}^T g_{\theta}(y_{k}|X_{k})\right].
$$
Our main objective,  which is assumed to be well-posed, is to find $\theta\in\Theta$ so as to maximize 
$p_{\theta}(y_{1},\dots,y_{T})$. The approach we will take is to obtain an expression for the score function,
$S(\theta):=\nabla_{\theta}\log\left(p_{\theta}(y_{1},\dots,y_{T})\right)$ on which to base an optimization procedure. 
We shall denote the assumed unique maximizer of $p_{\theta}(y_{1},\dots,y_{T})$ as $\theta_{\star}$.

\subsection{Bridging Approach}\label{sec:bridge_rev}

We review the method in \cite{schauer,vd_meulen_guided_mcmc} as has been used in \cite{beskos}.
Consider the case $t\in[s_1,s_2]$, $0\leq s_1,s_2\leq T$ and let
$\mathbf{X}_{[s_1,s_2]}:=\{X_{t}\}_{t\in[s_1,s_2]}$, and $\mathbf{W}_{[s_1,s_2]}:=\{W_{t}\}_{t\in[s_1,s_2]}$.  Set $f_{\theta,t,s_{2}}(x'|x)$ denote the unknown transition density from time $t$ to $s_2$ associated to \eqref{eq:diff}.
If one could sample from $f_{\theta,s_1,s_2}$ to obtain $(x, x')\in \mathbb{R}^{2d}$, we will explain that 
we can interpolate these points by using a bridge process.
This process will have
a drift given by $\mu_{\theta}(x)+\Sigma_{\theta}(x)\nabla_x\log{f}_{\theta,t,s_2}(x'|x)$.
Let $\mathbb{P}_{\theta,x,x'}$ denote the law of the solution of the SDE \eqref{eq:diff}, on $[s_1,s_2]$, started at $x$ and conditioned to hit $x'$ at time $s_2$.

We introduce a user-specified auxiliary process $\{\tilde{X}_t\}_{t\in[s_1,s_2]}$ following:
\begin{align}
\label{eq:aux_SDE_tilde}
d \tilde X_{t} = \tilde \mu_{\theta}(t,\tilde X_{t})dt + \tilde \sigma_{\theta}(t,\tilde X_{t})dW_{t}, \quad t\in[s_1,s_2],\quad~\tilde{X}_{s_1} =x, 
\end{align}
where for each $\theta\in\Theta$, $\tilde \mu_{\theta}:[s_1,s_2]\times\mathbb{R}^{d}\rightarrow\mathbb{R}^{d}$ and $\tilde \sigma_{\theta}:\mathbb{R}^{d}\rightarrow\mathbb{R}^{d\times d}$ is
such that for each $\theta\in\Theta$
$$
\tilde \Sigma_{\theta} (s_2,x'):= \tilde \sigma_{\theta}(s_2,x') \tilde \sigma_{\theta}(s_2,x')^{\top} \equiv \Sigma_{\theta}(x').
$$ 
\eqref{eq:aux_SDE_tilde} is specified so that its transition density $\tilde{f}_\theta$ is available; see \cite[Section 2.2]{schauer} for the technical conditions on $\tilde \mu_{\theta}, \tilde \Sigma_{\theta}, \tilde f_\theta$. 
The main purpose of $\{\tilde X_t\}_{t\in[s_1,s_2]}$ is to sample $x'$ and use its transition density to construct another process $\{X_t^\circ\}_{t\in[s_1,s_2]}$ conditioned to hit $x'$ at $t=s_2$; which in turn will be an importance proposal for $\{X_t\}_{t\in[s_1,s_2]}$.  Set:
\begin{align}
\label{eq:aux_SDE}
d X^\circ_{t} = \mu_{\theta,s_2}^{\circ}(t,X^{\circ}_{t}; x')dt + \sigma_{\theta}(X^{\circ}_{t})dW_{t}, \quad t\in[s_1,s_2],\quad~X^{\circ}_{s_1} =x, 
\end{align}
where:
$$
\mu_{\theta,s_2}^{\circ}(t,x;x')=\mu_{\theta}(x)+\Sigma_{\theta}(x)\nabla_x\log\tilde{f}_{\theta,t,s_2}(x'|x),
$$ 
and denote by
 $\mathbb{P}^\circ_{\theta,x,x'}$ the probability law of the solution of (\ref{eq:aux_SDE}). 
The SDE in \eqref{eq:aux_SDE} yields: 
\begin{equation}
\mathbf{W}\rightarrow C_{\theta,s_1,s_2}(x,\mathbf{W}_{[s_1,s_2]},x'),
\label{eq:map}
\end{equation} 
mapping the driving Wiener noise $\mathbf{W}$ to the solution of \eqref{eq:aux_SDE}, reparametering the problem from $\mathbf{X}$ to $(\mathbf{W},x')$.

\cite{schauer} prove that $\mathbb{P}_{\theta,x,x'}$ and $\mathbb{P}^\circ_{\theta,x,x'}$ are absolutely continuous w.r.t.~each other, with Radon-Nikodym derivative: 
\begin{equation}
\frac{d\mathbb{P}_{\theta,x,x'} }{d \mathbb{P}^\circ_{\theta,x,x'} }(\mathbf{X}_{[s_1,s_2]})=
\exp\Big\{  \int_{s_1}^{s_2} L_{\theta,s_2}(t,X_t)dt\Big\} \times \frac{ \tilde{f}_{\theta,s_1,s_2}(x'|x)}{f_{\theta,s_1,s_2}(x'|x)},
\label{eq:density}
\end{equation}
where: 
\begin{align*}
L_{\theta,s_2}&(t,x):=\left(\mu_\theta(x)- \tilde{\mu}_\theta(t,x)\right)^{\top}\, \nabla_x\log\tilde{f}_{\theta,t,s_2}(x'|x)\\
&-\frac{1}{2}\textrm{Tr}\,\Big\{\,\big[ \Sigma_{\theta}(x)-\tilde{\Sigma}_{\theta}(t,x)\big] \big[ -\nabla_x^2\log\tilde{f}_{\theta,t,s_2}(x'|x)-\nabla_x\log\tilde{f}_{\theta,t,s_2}(x'|x)\nabla_x\log\tilde{f}_{\theta,t,s_2}(x'|x)^{\top} \big]\,\Big\} ,         
\end{align*}
with $\textrm{Tr}(\cdot)$ denoting the trace of a squared matrix.

Now let $\overline{f}_{\theta,t,s_2}(x'|x)$ be any positive probability density on $\mathbb{R}^d$ and write expectations w.r.t.:
$$
\overline{\mathbb{P}}_{\theta}(d(x_0,\dots,x_T), d(\mathbf{w}_{[0,1]},\dots,\mathbf{w}_{[T-1,T]})) = \nu_{\theta}(x_0)dx_0\prod_{k=1}^T \mathbb{W}(d\mathbf{w}_{[k-1,k]})\overline{f}_{\theta,k-1,k}(x_k|x_{k-1})dx_k
$$
as $\overline{\mathbb{E}}_{\theta}[\cdot]$,  where $\mathbb{W}$ is the law of Brownian motion and for each $k\in\{0,\dots,T\}$, $dx_k$ is  $d-$dimensional Lebesgue measure.
Set
$$
R_{\theta,s_1,s_2}(\mathbf{X}_{[s_1,s_2]}) := \frac{d\mathbb{P}_{\theta,x,x'} }{d \mathbb{P}^\circ_{\theta,x,x'} }(\mathbf{X}_{[s_1,s_2]})\frac{f_{\theta,s_1,s_2}(x'|x)}{\overline{f}_{\theta,s_1,s_2}(x'|x)}.
$$
Then as a result of the above construction we have that 
$$
p_{\theta}(y_{1},\dots,y_{T}) = \overline{\mathbb{E}}_{\theta}\left[
\prod_{k=1}^T g_{\theta}(y_{k}|X_{k})R_{\theta,k-1,k}\left(
C_{\theta,k-1,k}(X_{k-1},\mathbf{W}_{[k-1,k]},X_k)\right)
\right].
$$
The score function $S(\theta)$ is equal to (and assumed to be finite for each $\theta\in\Theta$)
$$
S(\theta) =  
$$
\begin{equation}\label{eq:score}
\widehat{\mathbb{E}}_{\theta}\left[
\sum_{k=1}^T \nabla_{\theta}\log\left\{g_{\theta}(y_{k}|X_{k})R_{\theta,k-1,k}\left(
C_{\theta,k-1,k}(X_{k-1},\mathbf{W}_{[k-1,k]},X_k)\right)
\overline{f}_{\theta,k-1,k}(X_k|X_{k-1})
\right\}+\nabla_{\theta}\log(\nu_{\theta}(X_0))
\right]
\end{equation}
where we assume all derivatives are well-defined and $\widehat{\mathbb{E}}_{\theta}[\cdot]$ is an expectation w.r.t.:
$$
\widehat{\mathbb{P}}_{\theta}(d(x_0,\dots,x_T), d(\mathbf{w}_{[0,1]},\dots,\mathbf{w}_{[T-1,T]})) =
$$
$$ 
\frac{1}{p_{\theta}(y_{1},\dots,y_{T})}\left\{\prod_{k=1}^T g_{\theta}(y_{k}|x_{k})R_{\theta,k-1,k}\left(
C_{\theta,k-1,k}(x_{k-1},\mathbf{w}_{[k-1,k]},x_k)\right)
\right\}\overline{\mathbb{P}}_{\theta}(d(x_0,\dots,x_T), d(\mathbf{w}_{[0,1]},\dots,\mathbf{w}_{[T-1,T]})).
$$
See \cite{beskos,ub_grad} for similar identities.
\begin{rem}
In practice when trying to approximate $S(\theta)$ one will need to compute the $\theta$ derivative of $C_{\theta}$ which is often unavailable.  We will address this point when considering time discretization below.
\end{rem}

\subsection{Time Discretization}

In most cases working with the formulation \eqref{eq:aux_SDE}-\eqref{eq:density} to compute \eqref{eq:score}
is not possible.  As a result,  we now consider time discretization which will facilitate the application of practical
algorithms to compute the optimizer of the likelihood function $p_{\theta}(y_{1},\dots,y_{T})$.

We first give the standard Euler-Maruyama time discretization of the solution to \eqref{eq:aux_SDE} (associated to a time interval $[k-1,k]$, $k\in\{1,\dots,T\}$) on a regular grid of spacing $\Delta_{l}=2^{-l}$, with starting point $x_{k-1}^{\circ}$ and ending point $x_{k}^{\circ}$. That is for $j\in\{0,1\dots,\Delta_{l}^{-1}-2\}$:
\begin{align}\label{eq:disc_circ}
X_{k-1+(j+1)\Delta_{l}}^{\circ} & = X_{k-1+j\Delta_{l}}^{\circ} + 
\mu_{\theta,k}^{\circ}(k-1+j\Delta_{l},X^{\circ}_{k-1+j\Delta_{l}}; x_{k}^\circ)\Delta_{l} + 
\nonumber \\ &
\sigma_{\theta}(X^{\circ}_{k-1+j\Delta_{l}})\left[W_{k-1+(j+1)\Delta_{l}}-W_{k-1+j\Delta_{l}}\right].
\end{align}
Given $(x_{k-1}^{\circ},x_{k}^{\circ})$ and $\mathbf{W}_{[k-1,k]}^l=(W_{k-1+\Delta_{l}}-W_{k-1},\dots,W_{k-\Delta_l}-W_{k-2\Delta_{l}})$,  the recursion \eqref{eq:disc_circ} induces a discretized path 
$X_{k-1+\Delta_{l}}^{\circ},\dots,X_{k-\Delta_{l}}^{\circ}$ and we write such a path, including the starting and ending points with the notation
$$
C_{\theta,k-1,k}^l(x_{k-1}^{\circ},\mathbf{W}_{[k-1,k]}^l,x_{k}^{\circ}).
$$
We also need a discretization of the Radon-Nikodym derivative.  Consider 
$$
\mathbf{X}_{[k-1,k]}^l=(X_{k-1},X_{k-1+\Delta_{l}},\dots,X_k)
$$ 
then we set
$$
R_{\theta,k-1,k}^l(\mathbf{X}_{[k-1,k]}^l) := 
\exp\Big\{ \sum_{j=0}^{\Delta_{l}^{-1}-1} 
L_{\theta,k}(t,X_{k-1+j\Delta_{l}})\Delta_{l}\Big\} \times \frac{\tilde{f}_{\theta,k-1,k}(X_{k}|X_{k-1})}{\overline{f}_{\theta,k-1,k}(X_{k}|X_{k-1})}.
$$
Given this construction,  one induces the time-discretized score function:
$$
S^l(\theta)   
$$
\begin{equation}\label{eq:score_disc}
\widehat{\mathbb{E}}_{\theta}^l\left[
\sum_{k=1}^T \nabla_{\theta}\log\left\{g_{\theta}(y_{k}|X_{k})R_{\theta,k-1,k}^l\left(
C_{\theta,k-1,k}^l(X_{k-1},\mathbf{W}_{[k-1,k]}^l,X_k)\right)
\overline{f}_{\theta,k-1,k}(X_k|X_{k-1})
\right\}+\nabla_{\theta}\log(\nu_{\theta}(X_0))
\right]
\end{equation}
where $\widehat{\mathbb{E}}_{\theta}^l[\cdot]$ is an expectation w.r.t.:
$$
\widehat{\mathbb{P}}_{\theta}^l(d(x_0,\dots,x_T), d(\mathbf{w}_{[0,1]}^l,\dots,\mathbf{w}_{[T-1,T]}^l)) =
$$
$$
\frac{1}{p_{\theta}^l(y_{1},\dots,y_{T})}\left\{\prod_{k=1}^T g_{\theta}(y_{k}|x_{k})R_{\theta,k-1,k}^l\left(
C_{\theta,k-1,k}^l(x_{k-1},\mathbf{w}_{[k-1,k]}^l,x_k)\right)
\right\}\overline{\mathbb{P}}_{\theta}^l(d(x_0,\dots,x_T), d(\mathbf{w}_{[0,1]}^l,\dots,\mathbf{w}_{[T-1,T]}^l))
$$
with 
\begin{eqnarray*}
p_{\theta}^l(y_{1},\dots,y_{T}) & = & \int_{\mathbb{R}^{d(1+T\Delta_l^{-1})}} \left\{\prod_{k=1}^T g_{\theta}(y_{k}|X_{k})R_{\theta,k-1,k}^l\left(
C_{\theta,k-1,k}^l(x_{k-1},\mathbf{w}_{[k-1,k]}^l,x_k)\right)
\right\}
\times \\ & &
\overline{\mathbb{P}}_{\theta}^l(d(x_0,\dots,x_T), d(\mathbf{w}_{[0,1]}^l,\dots,\mathbf{w}_{[T-1,T]}^l)) \\
\overline{\mathbb{P}}_{\theta}^l(d(x_0,\dots,x_T), d(\mathbf{w}_{[0,1]}^l,\dots,\mathbf{w}_{[T-1,T]}^l))
& = & \nu_{\theta}(x_0)dx_0\prod_{k=1}^T \mathbb{W}(d\mathbf{w}_{[k-1,k]}^l)\overline{f}_{\theta,k-1,k}(x_k|x_{k-1})dx_k.
\end{eqnarray*}
\begin{rem}
To approximate $S^l(\theta)$ we will replace the $\theta$ derivative of $C_{\theta}^l$ with a finite difference
approximation,  whose convergence (step-size) is exactly $\Delta_l$.
\end{rem}
We shall denote the assumed unique maximizer of $p_{\theta}^l(y_{1},\dots,y_{T})$ as $\theta_{\star}^l$.
We assume explicitly that $\lim_{l\rightarrow\infty}\theta_{\star}^l=\theta_{\star}$, where we recall that the latter is the unique maximizer of $p_{\theta}(y_{1},\dots,y_{T})$.

\section{Methodology}\label{sec:method}

In order to describe our methodology we will begin in Section \ref{sec:ctmc} by constructing a Markov kernel with invariant measure
$\widehat{\mathbb{P}}_{\theta}$. Although it is not practically implementable,  it demonstrates that an associated Markov kernel with $\widehat{\mathbb{P}}_{\theta}^l$ as its invariant measure, is well-defined as the time discretization goes to zero. This and an associated coupled kernel are also described in Section \ref{sec:ctmc}.
Finally in Section \ref{sec:approach} we show how these kernels are used to obtain unbiased parameter estimates.
We emphasize that by unbiased we simply mean a stochastic estimator $\widehat{\theta}_{\star}$ such that 
$\mathbb{E}[\widehat{\theta}_{\star}]=\theta_{\star}$, where the expectation is w.r.t.~the law associated to the randomness used to construct the estimator and is conditioning on the fixed data set $(y_1,\dots,y_T)$.

\subsection{Markov Kernels}\label{sec:ctmc}

We begin by giving the first Markov kernel that has invariant measure $\widehat{\mathbb{P}}_{\theta}$. This is the CPF \cite{andrieu} with backward sampling \cite{lind,whiteley} and is presented in Algorithm \ref{alg:ctcpbs}.   In Algorithm \ref{alg:ctcpbs} the notation $\mathcal{C}\textrm{at}(\cdot)$ is used,  which denotes the categorical distribution on the indices $\{1,\dots,M\}$,  with input a probability vector of length $M$.  In Algorithm \ref{alg:ctcpbs} $M=N$ in each instance it is mentioned.

\subsubsection{Conditional Particle Filter with Backward Sampling}\label{sec:cpf}

For the next kernel we introduce some notations. $\mathbf{z}^l=(z_{1}^l,\dots,z_T^l)$
where $z_1^l=(x_0,x_1,\mathbf{w}_{[0,1]}^l)$ and for $k\in\{2,\dots,T\}$,  $z_k^l=(x_k,\mathbf{w}_{[k-1,k]}^l)$
The superscript indicates the level of time discretization and note that in the forthcoming exposition we may add a second superscript $\mathbf{z}^{n,l}$, $n$ will indicate the time of a sequence of random variables, i.e.~$\mathbf{z}^{0,l},\mathbf{z}^{1,l},\dots$.
A Markov kernel that has invariant measure $\widehat{\mathbb{P}}_{\theta}^l$ is presented in Algorithm \ref{alg:dtcpbs}.  It can be observed that the differences with Algorithm \ref{alg:ctcpbs} are minimal,  with the exception that Algorithm \ref{alg:dtcpbs} can be run in practice whereas Algorithm \ref{alg:ctcpbs} is impossible to run on a computer. In Algorithm \ref{alg:dtcpbs} we start with a given $\mathbf{z}^l$
and, for a given $\theta$,  produce a new  $\mathbf{z}^l=(z_{1}^l,\dots,z_T^l)$ and this Markov kernel has exactly 
$\widehat{\mathbb{P}}_{\theta}^l$ as its invariant measure.  We call this Markov kernel $K_{\theta,l}(\mathbf{z}^{n-1,l},d\mathbf{z}^{n,l})$, $n\in\mathbb{N}$.  

We remark that the time discretization of Algorithm \ref{alg:ctcpbs} into Algorithm \ref{alg:dtcpbs} is rather important as it shows there is a well-defined algorithm (Algorithm \ref{alg:ctcpbs}) as $l$ the level of time discretization grows in the case of Algorithm \ref{alg:dtcpbs}.
In the approaches of \cite{ub_par,ub_grad} which rely on the Girsanov theorem with simple Euler-Maruyama time discretization,  the option of backward sampling is feasible, but destined to fail. The reason being that the backward weights that are used in those cases would have an exponential variance in $l$ rendering the methodology worthless in the cases where one would want to use them, i.e.~for the most precise time discretization. As a result,  the potential gains in convergence speed for backward sampling are lost for those afore-mentioned approaches wheras this is not the case here.  

We can also contrast Algorithm \ref{alg:dtcpbs} with the case of removing the backward sampling.  This would replace Step 5.~of Algorithm \ref{alg:dtcpbs} with:
\begin{itemize}
\item{Step \textbf{(F)}: Set $j_k^l = a_k^{j_{k+1}^l,l}$ for $k\in\{1,\dots,T-1\}$.
}
\end{itemize}
Using Step \textbf{(F)} instead of Step 5.~of Algorithm \ref{alg:dtcpbs} is expected to yield a more poorly performing algorithm, but provides a benchmark for the contribution in this article.

\subsubsection{Coupled Conditional Particle Filter with Backward Sampling}\label{sec:ccpf}

To define the last kernel that we will use for our method we need to recall the simulation of a maximal coupling of two probability mass functions with the same support. This is given in Algorithm \ref{alg:max_coup}.
For two positive probability mass functions on $\{1,\dots,N\}$ call them $(r_1^1,\dots,r_1^N)$ and $(r_2^1,\dots,r_2^N)$ as in Algorithm \ref{alg:max_coup} we denote the probability distribution on $\{1,\dots,N\}^2$ which corresponds to the maximal coupling as 
$$
\mathcal{M}\textrm{ax}\left(\left\{r_1^1,\dots,r_1^N\right\},\left\{r_2^1,\dots,r_2^N\right\}\right).
$$

To complete the specification for our last Markov kernel, we need to introduce several couplings as we now describe.  First, for each $\phi=(\theta,\theta')\in\Theta^2$, let $\check{\nu}_{\phi}$ be a coupling of $(\nu_{\theta},\nu_{\theta'})$, that is $\check{\nu}_{\phi}$ is a probability on $\mathbb{R}^{2d}$ such that for any Borel set $A\subseteq\mathbb{R}^d$, $\int_{A}\nu_{\theta}(x_0)dx_0 = \int_{A\times\mathbb{R}^d}\check{\nu}_{\phi}(d(x_0,x_0'))$ and $\int_{A}\nu_{\theta'}(x_0')dx_0' = \int_{\mathbb{R}^d\times A}\check{\nu}_{\phi}(d(x_0,x_0'))$.  Second,  for each $\phi=(\theta,\theta')\in\Theta^2$, $k\in\{1,\dots,T\}$ and any $(x,x')\in\mathbb{R}^{2d}$ let $\check{f}_{\phi,k-1,k}(d(v,v')|x,x')$ be a coupling of $(\overline{f}_{\theta,k-1,k}(v|x)dv,\overline{f}_{\theta',k-1,k}(v'|x')dv')$,  that is for any  Borel set $A\subseteq\mathbb{R}^d$,
$$
\int_{A\times\mathbb{R}^d}\check{f}_{\phi,k-1,k}(d(v,v')|x,x') = \int_A \overline{f}_{\theta,k-1,k}(v|x)dv
\quad\textrm{and}\quad \int_{\mathbb{R}^d\times A}\check{f}_{\phi,k-1,k}(d(v,v')|x,x') = \int_A \overline{f}_{\theta',k-1,k}(v'|x')dv'.
$$
Finally for $k\in\{1,\dots,T\}$, $l\in\mathbb{N}$ and $\mathbf{w}_{[k-1,k]}^l\in\mathbb{R}^{d\Delta_l^{-1}}$
define the function $\mathsf{C}^l:\mathbb{R}^{d\Delta_l^{-1}}\rightarrow\mathbb{R}^{d\Delta_{l-1}^{-1}}$
as
$$
\mathsf{C}^l(\mathbf{w}_{[k-1,k]}^l) = \left(w_{k-1+\Delta_l}-w_{k-1}+w_{k-1+2\Delta_l}-w_{k-1+\Delta_l},\dots,
w_{k-\Delta_l}-w_{k-2\Delta_l}+w_{k}-w_{k-\Delta_l}
\right).
$$
Then we define the coupling of Brownian motions $\check{\mathbb{W}}^l$ as the probability on 
$\mathbb{R}^{d(\Delta_l^{-1}+\Delta_{l-1}^{-1})}$ as
$$
\check{\mathbb{W}}^l(d(\mathbf{w}_{[k-1,k]}^l,\mathbf{w}_{[k-1,k]}^{l-1})) = \mathbb{W}(d\mathbf{w}_{[k-1,k]}^l)
\delta_{\{\mathsf{C}^l(\mathbf{w}_{[k-1,k]}^l)\}}(d\mathbf{w}_{[k-1,k]}^{l-1})
$$
where $\delta_{\{\mathsf{C}^l(\mathbf{w}_{[k-1,k]}^l)\}}$ is the Dirac measure on the set $\{\mathsf{C}^l(\mathbf{w}_{[k-1,k]}^l)\}$. 

Given the above notations we present Algorithm \ref{alg:dtccpbs}.  This algorithm will take in as its input
$\mathbf{z}^{n-1,l}=(z_{1}^{n-1,l},\dots,z_T^{n-1,})$ and $\bar{\mathbf{z}}^{n-1,l-1}=(\bar{z}_{1}^{n-1,l-1},\dots,\bar{z}_T^{n-1,l-1})$ where $n\in\mathbb{N}$,  
$\bar{z}_1^{n-1,l-1}=(\bar{x}_0^{n-1},\bar{x}_1^{n-1},\mathbf{w}_{[0,1]}^{n-1,l-1})$ and for $k\in\{2,\dots,T\}$,  $\bar{z}_k^{n-1,l-1}=(\bar{x}_k^{n-1},\mathbf{w}_{[k-1,k]}^{n-1,l-1})$ and for the given $\phi=(\theta,\theta')$ produce $(\mathbf{z}^{n,l},\bar{\mathbf{z}}^{n,l-1})$ .  Moreover,  if one considers only the transition
of $\mathbf{z}^{n-1,l}$ to $\mathbf{z}^{n,l}$ it is as if it has been sampled by 
$K_{\theta,l}(\mathbf{z}^{n-1,l},d\mathbf{z}^{n,l})$. Similarly the 
transition from $\bar{\mathbf{z}}^{n-1,l-1}$ 
to $\bar{\mathbf{z}}^{n,l-1}$
is as if it has been sampled from
$K_{\theta',l-1}(\bar{\mathbf{z}}^{n-1,l-1},d\bar{\mathbf{z}}^{n,l-1})$.  The key point is that these Markov kernels have been coupled
and this is important if one seeks to show that our final algorithm will produce unbiased estimates.
We call this Markov kernel $\check{K}_{\phi,l}(\mathbf{v}^{n-1,l},d\mathbf{v}^{n,l})$ where $\mathbf{v}^{n,l}=(\mathbf{z}^{n,l},\bar{\mathbf{z}}^{n,l-1})$.

As for the case of the CPF,  we can compare Algorithm \ref{alg:dtccpbs} with the scenario where one removes backward sampling.  This would replace Step 5.~of Algorithm \ref{alg:dtccpbs} with:
\begin{itemize}
\item{Step \textbf{(FC)}: Set $j_k^l = a_k^{j_{k+1}^l,l}$,   $\bar{j}_k^{l-1} = a_k^{\bar{j}_{k+1}^{l-1},l-1}$ for $k\in\{1,\dots,T-1\}$.
}
\end{itemize}
Again,  we expect this approach to be less efficient than using Algorithm \ref{alg:dtccpbs},  but provides a reference for comparison for our methodology.

\begin{algorithm}[h!]
\caption{Continuous Time Conditional Backward Sampling Particle Filter}
\label{alg:ctcpbs}
\begin{enumerate}
\item{Input: $(x_0,\dots,x_T)$, $(\mathbf{w}_{[0,1]},\dots,\mathbf{w}_{[T-1,T]})$, $\theta\in\Theta$ and $N\in\mathbb{N}$. Go to step 2..}
\item{Forward Initialize:   For $i\in\{1,\dots,N-1\}$
independently sample $\xi_1^i=(x_0^i,x_1^i,\mathbf{w}_{[0,1]}^i)$ from $\nu_{\theta}(x_0)dx_0
\overline{f}_{\theta}(x_1|x_0)dx_1\mathbb{W}(d\mathbf{w}_{[0,1]})$. 
Set $\xi_1^N=(x_0,x_1,\mathbf{w}_{[0,1]})$, 
$k=1$ and go to step 3..}
\item{Iterate Forwards: For $i\in\{1,\dots,N\}$ compute
$$
\alpha_k^i = g_{\theta}(y_k|x_k^i)R_{\theta,k-1,k}(C_{\theta,k-1,k}(\xi_{k}^i)).
$$
For $i\in\{1,\dots,N-1\}$ independently sample $a_k^i\in\{1,\dots,N\}$ from 
$$
\mathcal{C}\textrm{at}\left(\frac{\alpha_k^1}{\sum_{s=1}^N\alpha_k^s},\dots,\frac{\alpha_k^N}{\sum_{s=1}^N\alpha_k^s}\right)
$$
and $\check{x}_{k}^i= x_k^{a_k^i}$; set $a_k^N=N$.
Then, for $i\in\{1,\dots,N-1\}$, sample $x_{k+1}^i,\mathbf{w}_{[k,k+1]}^i$ from $\overline{f}_{\theta}(x_{k+1}^i|\check{x}_k^i)dx_{k+1}^i\mathbb{W}(d\mathbf{w}_{[k,k+1]}^i)$,  with $\xi_{k+1}^i = (\check{x}_k^i,\mathbf{w}_{[k,k+1]}^i,x_{k+1}^i)$.
Set $\xi_{k+1}^N=(x_{k},\mathbf{w}_{[k,k+1]},x_{k+1})$ and $k=k+1$. If $k=T$ 
go to step 4.~otherwise return to the start of step 3..}
\item{Backwards Initialize: For $i\in\{1,\dots,N\}$ compute
$$
\alpha_T^i = g_{\theta}(y_T|x_T^i)R_{\theta,T-1,T}(C_{\theta,T-1,T}(\xi_{T}^i)).
$$
Sample $j_T\in\{1,\dots,N\}$ via $\mathcal{C}\textrm{at}\left(\tfrac{\alpha_T^1}{\sum_{s=1}^N\alpha_T^s},\dots,\tfrac{\alpha_T^N}{\sum_{s=1}^N\alpha_T^s}\right).$  Set $k=T-1$ and go to step 5..
}
\item{Iterate Backwards: For $i\in\{1,\dots,N\}$ compute 
$$
\beta_k^i = \alpha_k^i \overline{f}_{\theta}(x_{k+1}^{j_{k+1}}|x_k^{i})
R_{\theta,k,k+1}(C_{\theta,k-1,k}(x_k^i,\mathbf{w}_{[k,k+1]}^{j_{k+1}},x_{k+1}^{j_{k+1}})).
$$
Sample $j_k\in\{1,\dots,N\}$ via $\mathcal{C}\textrm{at}\left(\tfrac{\beta_k^1}{\sum_{s=1}^N\beta_k^s},\dots,\tfrac{\beta_k^N}{\sum_{s=1}^N\beta_k^s}\right).$  Set $k=k-1$ and if $k=-1$ go to step 6.~otherwise go to the start of step 5.
}
\item{Output: $(x_0^{j_1},x_1^{j_1},\dots,x_T^{j_T})$, $(\mathbf{w}_{[0,1]}^{j_1},\dots,\mathbf{w}_{[T-1,T]}^{j_T})$.}
\end{enumerate}
\end{algorithm}

\begin{algorithm}[h!]
\caption{Time Discretized Conditional Backward Sampling Particle Filter}
\label{alg:dtcpbs}
\begin{enumerate}
\item{Input: 
$l\in\mathbb{N}\cup\{0\}$,  
$(x_0,\dots,x_T)$, $(\mathbf{w}_{[0,1]}^l,\dots,\mathbf{w}_{[T-1,T]}^l)$, $\theta\in\Theta$  and $N\in\mathbb{N}$. Go to step 2.}
\item{Forward Initialize:  For $i\in\{1,\dots,N-1\}$
independently sample $\xi_1^{i,l}=(x_0^{i,l},x_1^{i,l},\mathbf{w}_{[0,1]}^{i,l})$ from $\nu_{\theta}(x_0)dx_0
\overline{f}_{\theta}(x_1|x_0)dx_1\mathbb{W}(d\mathbf{w}_{[0,1]}^l)$. 
Set $\xi_1^{N,l}=(x_0,x_1,\mathbf{w}_{[0,1]}^l)$, $k=1$ and go to step 3.}
\item{Iterate Forwards: For $i\in\{1,\dots,N\}$ compute
$$
\alpha_k^{i,l} = g_{\theta}(y_k|x_k^{i,l})R_{\theta,k-1,k}^l(C_{\theta,k-1,k}^l(\xi_{k}^{i,l})).
$$
For $i\in\{1,\dots,N-1\}$ independently sample $a_k^{i,l}\in\{1,\dots,N\}$ from
$$
\mathcal{C}\textrm{at}\left(\frac{\alpha_k^{1,l}}{\sum_{s=1}^N\alpha_k^{s,l}},\dots,\frac{\alpha_k^{N,l}}{\sum_{s=1}^N\alpha_k^{s,l}}\right)
$$
and $\check{x}_{k}^{i,l}= x_k^{a_k^{i,l},l}$; set $a_k^{N,l}=N$.
Then, for $i\in\{1,\dots,N-1\}$, sample $x_{k+1}^{i,l},\mathbf{w}_{[k,k+1]}^{i,l}$ from $\overline{f}_{\theta}(x_{k+1}^{i,l}|\check{x}_k^{i,l})dx_{k+1}^{i,l}\mathbb{W}(d\mathbf{w}_{[k,k+1]}^{i,l})$,  with $\xi_{k+1}^{i,l} = (\check{x}_k^{i,l},\mathbf{w}_{[k,k+1]}^{i,l},x_{k+1}^{i,l})$.
Set $\xi_{k+1}^{1,l}=(x_{k},\mathbf{w}_{[k,k+1]},x_{k+1})$ and $k=k+1$.  If $k=T$ 
go to step 4.~otherwise return to the start of step 3.}
\item{Backwards Initialize: For $i\in\{1,\dots,N\}$ compute
$$
\alpha_T^{i,l} = g_{\theta}(y_T|x_T^{i,l})R_{\theta,k-1,k}^l(C_{\theta,k-1,k}^l(\xi_{T}^{i,l})).
$$
Sample $j_T^l\in\{1,\dots,N\}$ via $\mathcal{C}\textrm{at}\left(\tfrac{\alpha_T^{1,l}}{\sum_{s=1}^N\alpha_T^{s,l}},\dots,\tfrac{\alpha_T^{N,l}}{\sum_{s=1}^N\alpha_T^{s,l}}\right).$  Set $k=T-1$ and go to step 5.
}
\item{Iterate Backwards: For $i\in\{1,\dots,N\}$ compute 
$$
\beta_k^{i,l} = \alpha_k^{i,l} \overline{f}_{\theta}(x_{k+1}^{j_{k+1}^l,l}|x_k^{i,l})
R_{\theta,k,k+1}^l(C_{\theta,k-1,k}^l(x_k^{i,l},\mathbf{w}_{[k,k+1]}^{j_{k+1}^l,l},x_{k+1}^{j_{k+1}^l,l})).
$$
Sample $j_k^l\in\{1,\dots,N\}$ via $\mathcal{C}\textrm{at}\left(\tfrac{\beta_k^{1,l}}{\sum_{s=1}^N\beta_k^{s,l}},\dots,\tfrac{\beta_k^{N,l}}{\sum_{s=1}^N\beta_k^{s,l}}\right).$  Set $k=k-1$ and if $k=-1$ go to step 6.~otherwise go to the start of step 5.
}
\item{Output: $(x_0^{j_1^l,l},x_1^{j_1^l,l},\dots,x_T^{j_T^l,l})$, $(\mathbf{w}_{[0,1]}^{j_1^l,l},\dots,\mathbf{w}_{[T-1,T]}^{j_T^l,l})$.}
\end{enumerate}
\end{algorithm}

\begin{algorithm}[h!]
\begin{enumerate}
\item{Input: Two probability mass functions (PMFs) $(r_1^1,\dots,r_1^N)$ and $(r_2^1,\dots,r_2^N)$ on $\{1,\dots,N\}$.}
\item{Generate $U\sim\mathcal{U}_{[0,1]}$ (uniform distribution on $[0,1]$).}
\item{If $U<\sum_{i=1}^N \min\{r_1^i,r_2^i\}=:\bar{r}$ then generate $i\in\{1,\dots,N\}$ according to the probability mass function:
$$
r_3^i = \frac{1}{\bar{r}} \min\{r_1^i,r_2^i\}
$$
and set $j=i$.}
\item{Otherwise generate $i\in\{1,\dots,N\}$ and $j\in\{1,\dots,N\}$ independently according to the probability mass functions 
$$
r_4^i = \frac{1}{1-\bar{r}} (r_1^i - \min\{r_1^i,r_2^i\})
$$
and
$$
r_5^j = \frac{1}{1-\bar{r}} (r_2^j - \min\{r_1^j,r_2^j\})
$$
respectively.
}
\item{Output: $(i,j)\in\{1,\dots,N\}^2$. $i$, marginally has PMF $r_1^i$ and $j$, marginally has PMF $r_2^j$.}
\end{enumerate}
\caption{Simulating a Maximal Coupling of Two Probability Mass Functions on $\{1,\dots,N\}$.}
\label{alg:max_coup}
\end{algorithm}

\begin{algorithm}[h]
\caption{Time Discretized Coupled Conditional Backward Sampling Particle Filter}
\label{alg:dtccpbs}
\begin{enumerate}
\item{Input: 
$l\in\mathbb{N}$, 
$(x_0^l,\dots,x_T^l)$, $(\mathbf{w}_{[0,1]}^l,\dots,\mathbf{w}_{[T-1,T]}^l)$, 
$(x_0^{l-1},\dots,x_T^{l-1})$, $(\mathbf{w}_{[0,1]}^{l-1},\dots,\mathbf{w}_{[T-1,T]}^{l-1})$,
$\phi=(\theta^l,\bar{\theta}^{l-1})\in\Theta^2$  and $N\in\mathbb{N}$. Go to step 2.}
\item{Forward Initialize:  For $i\in\{1,\dots,N-1\}$
independently sample $(\xi_1^{i,l},\bar{\xi}_1^{i,l-1})=
\left((x_0^{i,l},x_1^{i,l},\mathbf{w}_{[0,1]}^{i,l}),(\bar{x}_0^{i,l-1},\bar{x}_1^{i,l-1},\bar{\mathbf{w}}_{[0,1]}^{i,l-1})\right)$ from 
$$
\check{\nu}_{\phi}(d(x_0,x_0'))
\check{f}_{\phi}(d(x_1,x_1')|x_0,x_0')\check{\mathbb{W}}^l(d(\mathbf{w}_{[0,1]}^l,\mathbf{w}_{[0,1]}^{l-1})).
$$
Set
$\xi_1^{N,l}=(x_0^l,x_1^l,\mathbf{w}_{[0,1]}^l)$, $\bar{\xi}_1^{N,l-1}=(x_0^{l-1},x_1^{l-1},\mathbf{w}_{[0,1]}^{l-1})$,  $k=1$ and go to step 3.}
\item{Iterate Forwards: For $i\in\{1,\dots,N\}$ compute
$$
\alpha_k^{i,l} = g_{\theta^l}(y_k|x_k^{i,l})R_{\theta^l,k-1,k}^l(C_{\theta^l,k-1,k}^l(\xi_{k}^{i,l})) \quad
\bar{\alpha}_k^{i,l-1} = g_{\bar{\theta}^{l-1}}(y_k|\bar{x}_k^{i,l-1})R_{\bar{\theta}^{l-1},k-1,k}^{l-1}(C_{\bar{\theta}^{l-1},k-1,k}^{l-1}(\bar{\xi}_{k}^{i,l-1})).
$$
For $i\in\{1,\dots,N-1\}$ independently sample $(a_k^{i,l},\bar{a}_k^{i,l-1})\in\{1,\dots,N\}^2$ from 
$$
\mathcal{M}\textrm{ax}\left(\left\{\frac{\alpha_k^{1,l}}{\sum_{s=1}^N\alpha_k^{s,l}},\dots,\frac{\alpha_k^{N,l}}{\sum_{s=1}^N\alpha_k^{s,l}}\right\}, \left\{
\frac{\bar{\alpha}_k^{1,l-1}}{\sum_{s=1}^N\bar{\alpha}_k^{s,l-1}},\dots,\frac{\bar{\alpha}_k^{N,l-1}}{\sum_{s=1}^N\bar{\alpha}_k^{s,l-1}}
\right\}
\right)
$$
and $(\check{x}_{k}^{i,l},\check{\bar{x}}_{k}^{i,l-1})=(x_k^{a_k^{i,l},l},\bar{x}_k^{\bar{a}_k^{i,l-1},l-1})$;
set $(a_k^{N,l},\bar{a}_k^{N,l-1})=(N,N)$.
For $i\in\{1,\dots,N-1\}$, sample $\left((x_{k+1}^{i,l},\mathbf{w}_{[k,k+1]}^{i,l}),(\bar{x}_{k+1}^{i,l-1},\bar{\mathbf{w}}_{[k,k+1]}^{i,l})\right)$ from 
$$
\check{f}_{\phi}(d(x_{k+1}^{i,l},\bar{x}_{k+1}^{i,l-1})|(\check{x}_{k}^{i,l},\check{\bar{x}}_{k}^{i,l-1}))
\check{\mathbb{W}}^l(d(\mathbf{w}_{[k,k+1]}^{i,l},\mathbf{w}_{[k,k+1]}^{i,l-1}))
$$
with $\xi_{k+1}^{i,l} = (\check{x}_k^{i,l},\mathbf{w}_{[k,k+1]}^{i,l},x_{k+1}^{i,l})$,  $\bar{\xi}_{k+1}^{i,l-1} = (\check{\bar{x}}_k^{i,l-1},\bar{\mathbf{w}}_{[k,k+1]}^{i,l-1},\bar{x}_{k+1}^{i,l-1})$.
Set $\xi_{k+1}^{N,l}=(x_{k},\mathbf{w}_{[k,k+1]}^l,x_{k+1})$,  $\bar{\xi}_{k+1}^{N,l-1}=(x_{k}^{l-1},\mathbf{w}_{[k,k+1]}^{l-1},x_{k+1}^{l-1})$ and $k=k+1$. 
If $k=T$ 
go to step 4.~otherwise return to the start of step 3.}
\item{Backwards Initialize: For $i\in\{1,\dots,N\}$ compute
$$
\alpha_T^{i,l} = g_{\theta^l}(y_T|x_T^{i,l})R_{\theta^l,T-1,T}^l(C_{\theta^l,T-1,T}^l(\xi_{T}^{i,l})) \quad
\bar{\alpha}_T^{i,l-1} = g_{\bar{\theta}^{l-1}}(y_T|\bar{x}_T^{i,l-1})R_{\bar{\theta}^{l-1},T-1,T}^{l-1}(C_{\bar{\theta}^{l-1},T-1,T}^{l-1}(\bar{\xi}_{T}^{i,l-1})).
$$
Sample $(j_T^l,\bar{j}_T^{l-1})\in\{1,\dots,N\}^2$ via
$$
\mathcal{M}\textrm{ax}\left(\left\{\frac{\alpha_T^{1,l}}{\sum_{s=1}^N\alpha_T^{s,l}},\dots,\frac{\alpha_T^{N,l}}{\sum_{s=1}^N\alpha_T^{s,l}}\right\}, \left\{
\frac{\bar{\alpha}_T^{1,l-1}}{\sum_{s=1}^N\bar{\alpha}_T^{s,l-1}},\dots,\frac{\bar{\alpha}_T^{N,l-1}}{\sum_{s=1}^N\bar{\alpha}_T^{s,l-1}}
\right\}
\right).
$$
Set $k=T-1$ and go to step 5.
}
\item{Iterate Backwards: For $i\in\{1,\dots,N\}$ compute 
\begin{eqnarray*}
\beta_k^{i,l} & = & \alpha_k^{i,l} \overline{f}_{\theta^l}(x_{k+1}^{j_{k+1}^l,l}|x_k^{i,l})
R_{\theta^l,k,k+1}^l(C_{\theta^l,k-1,k}(x_k^{i,l},\mathbf{w}_{[k,k+1]}^{j_{k+1}^l,l},x_{k+1}^{j_{k+1}^l,l}))
\\
\bar{\beta}_k^{i,l-1} & = & \bar{\alpha}_k^{i,l-1} \overline{f}_{\bar{\theta}^{l-1}}(x_{k+1}^{\bar{j}_{k+1}^{l-1},l-1}|\bar{x}_k^{i,l-1})R_{\bar{\theta}^{l-1},k,k+1}^{l-1}(C_{\bar{\theta}^{l-1},k-1,k}^{l-1}(\bar{x}_k^{i,l-1},\bar{\mathbf{w}}_{[k,k+1]}^{\bar{j}_{k+1}^{l-1},l-1},\bar{x}_{k+1}^{\bar{j}_{k+1}^{l-1},l-1})).
\end{eqnarray*}
Sample $(j_k^l,\bar{j}_k^{l-1})\in\{1,\dots,N\}^2$ via
$$
\mathcal{M}\textrm{ax}\left(
\left\{\frac{\beta_k^{1,l}}{\sum_{s=1}^N\beta_k^{s,l}},\dots,\frac{\beta_k^{N,l}}{\sum_{s=1}^N\beta_k^{s,l}}\right\},
\left\{
\frac{\bar{\beta}_k^{1,l-1}}{\sum_{s=1}^N\bar{\beta}_k^{s,l-1}},\dots,\frac{\bar{\beta}_k^{N,l-1}}{\sum_{s=1}^N\bar{\beta}_k^{s,l-1}}
\right\}
\right).
$$
Set $k=k-1$ and if $k=-1$ go to step 6.~otherwise go to the start of step 5.
}
\item{Output: $(x_0^{j_1^l,l},x_1^{j_2^l,l},\dots,x_T^{j_T^l,l})$, $(\mathbf{w}_{[0,1]}^{j_1^l,l},\dots,\mathbf{w}_{[T-1,T]}^{j_T^l,l})$,  $(\bar{x}_0^{\bar{j}_1^{l-1},l-1},\bar{x}_1^{\bar{j}_2^{l-1},l-1},\dots,\bar{x}_T^{\bar{j}_T^{l-1},l-1})$, $(\bar{\mathbf{w}}_{[0,1]}^{\bar{j}_1^{l-1},l-1},\dots,\bar{\mathbf{w}}_{[T-1,T]}^{\bar{j}_T^{l-1},l-1})$.}
\end{enumerate}
\end{algorithm}

\subsection{Final Approach}\label{sec:approach}

The approach that we will use to obtain parameter estimates is that in \cite{ub_par},  except adapted to the more general case of this article and with the Markov kernels that we have developed in the previous section. 
To make the notations more succinct we set for $(\theta,l,\mathbf{z}^l)$ given
$$
H_l(\theta,\mathbf{z}^l) := \sum_{k=1}^T \nabla_{\theta}\log\left\{g_{\theta}(y_{k}|x_{k})R_{\theta,k-1,k}^l\left(
C_{\theta,k-1,k}^l(x_{k-1},\mathbf{w}_{[k-1,k]}^l,x_k)\right)\right\}+\nabla_{\theta}\log(\nu_{\theta}(x_0)).
$$
Also set 
\begin{eqnarray}
\check{\mathbb{P}}_{\theta}^l(d\mathbf{v}^l) & = & \check{\nu}_\theta(d(x_0^l,\bar{x}_0^{l-1}))
\check{f}_{\theta,0,1}(d(x_1^l,\bar{x}_1^{l-1})|x_0^l,\bar{x}_0^{l-1})\check{\mathbb{W}}^l(d(\mathbf{w}_{[0,1]}^l,\mathbf{w}_{[0,1]}^{l-1})) \times \nonumber\\ & &
\prod_{k=2}^T 
\check{f}_{\theta,k-1,k}(d(x_k^l,\bar{x}_k^{l-1})|x_{k-1}^l,\bar{x}_{k-1}^{l-1})\check{\mathbb{W}}^l(d(\mathbf{w}_{[k-1,k]}^l,\mathbf{w}_{[k-1,k]}^{l-1})).\label{eq:p_initial_coup}
\end{eqnarray}

The approach we use is to run a Markovian stochastic approximation \cite{andr} for a random time iteration and
for a random time-discretization level.  The approach is based upon the randomization schemes; see for instance
\cite{rhee,matti}.
To that end let $\mathbb{P}_{\mathtt{L}}(l)$ be a positive probability on $\mathbb{N}_0=\mathbb{N}\cup\{0\}$ which will be used to simulate a level of time discretization.   Let $(N_p)_{p\in\mathbb{N}_0}$ be
any sequence of increasing natural numbers, converging to infinity. Let $\mathbb{P}_{\mathtt{P}}$ be any positive probability on $\mathbb{N}_0$.  The sequence $(N_p)_{p\in\mathbb{N}_0}$ is used as the number of iterations and 
$\mathbb{P}_{\mathtt{P}}$ is used as probability distribution so as to select that iteration number.
Our Unbiased Markovian Stochastic Approximation (UBMSA) method is presented in Algorithm \ref{alg:USMA}.
Note that typically this algorithm is run $M\in\mathbb{N}$ times in parallel and the results averaged to obtain a final estimator.

UBMSA as given in Algorithm \ref{alg:USMA} has a large number of simulation parameters and settings to make.
This has been explored and discussed extensively in \cite{ub_par} and we give recommendations in Section \ref{sec:theory}.

\begin{algorithm}[h]
\caption{Unbiased Markovian Stochastic Approximation (UMSA)}
\label{alg:USMA}
\begin{algorithmic}[1]
\item{Sample $l$ from $\mathbb{P}_{\mathtt{L}}$ and $p$ from $\mathbb{P}_{\mathtt{P}}$.}
\item If $l=0$ perform the following:
\begin{itemize}
\item{Set $\theta_0^l\in\Theta$, $n=1$ and generate $\mathbf{z}^{n-1,l}\sim\overline{\mathbb{P}}_{\theta_0}^l$.}
\item{Using Algorithm \ref{alg:dtcpbs} sample $\mathbf{z}^{n,l}|(\theta_0^l,\mathbf{z}^{0,l}),\dots,(\theta_{n-1}^l,\mathbf{z}^{n-1,l})$ from $K_{\theta_{n-1}^l,l}(\mathbf{z}^{n-1,l},\cdot)$.}
\item{Update:

\begin{align}
\theta_n^l = \theta_{n-1}^l + \gamma_n H_l(\theta_{n-1}^l,\mathbf{z}^{n,l}).
\label{eq:sing_SA}
\end{align}

Set $n=n+1$. If $n=N_p+1$ go to the next bullet point, otherwise go back to the second bullet point.}
\item{If $p=0$ return
$$
\widehat{\theta}_{\star} = \frac{\theta_{N_p}^l}{\mathbb{P}_{\mathtt{P}}(p)\mathbb{P}_{\mathtt{L}}(l)},
$$
otherwise  return
$$
\widehat{\theta}_{\star} = \frac{\theta_{N_p}^l-\theta_{N_{p-1}}^l}{\mathbb{P}_{\mathtt{P}}(p)\mathbb{P}_{\mathtt{L}}(l)}.
$$
}
\end{itemize}
\item Otherwise perform the following:
\begin{itemize}
\item Set $\theta_0^l=\theta_0^{l-1}\in\Theta$, $n=1$ and generate $\mathbf{v}^{0,l}\sim\check{\mathbb{P}}_{\theta_0^l}^l$.
\item Using Algorithm \ref{alg:dtccpbs} sample $\mathbf{v}^{n,l}\Big|(\theta_0^l,\theta_0^{l-1},\mathbf{v}^{0,l}),\dots,(\theta_{n-1}^l,\theta_{n-1}^{l-1},\mathbf{v}^{n-1,l})$ from
$\check{K}_{\phi_{n-1}^l,l}\left(\mathbf{v}^{n-1,l},\cdot\right)$, where $\phi_{n-1}^l=(\theta_{n-1}^l,\theta_{n-1}^{l-1})$.

\item Update:
\begin{eqnarray*}
\theta_n^l & = & \theta_{n-1}^l + \gamma_n H_l(\theta_{n-1}^l,\mathbf{z}^{n,l}), \\
\theta_n^{l-1} & = & \theta_{n-1}^{l-1} + \gamma_n H_{l-1}(\theta_{n-1}^{l-1},\bar{\mathbf{z}}^{n,l-1}).
\end{eqnarray*}
Set $n=n+1$ and if $n=N_p+1$ go to the next bullet point, otherwise go back to the second bullet point.
\item If $p=0$ return
$$
\widehat{\theta}_{\star} = \frac{\theta_{N_p}^l-\theta_{N_p}^{l-1}}{\mathbb{P}_{\mathtt{P}}(p)\mathbb{P}_{\mathtt{L}}(l)},
$$
otherwise  return
$$
\widehat{\theta}_{\star} = \frac{\theta_{N_p}^l-\theta_{N_{p}}^{l-1}-\{\theta_{N_{p-1}}^l-\theta_{N_{p-1}}^{l-1}\}}{\mathbb{P}_{\mathtt{P}}(p)\mathbb{P}_{\mathtt{L}}(l)}.
$$
\end{itemize}
\end{algorithmic}
\end{algorithm}

\section{Theory}\label{sec:theory}

Our analysis will extend to the case where we replace Step 5.~of Algorithm \ref{alg:dtcpbs} with Step \textbf{(F)}
(given in Section \ref{sec:cpf}) and we replace Step 5.~of Algorithm \ref{alg:dtccpbs} with Step \textbf{(FC)} (given in Section \ref{sec:ccpf}).

\subsection{Assumptions} 

In our analysis we will make the following assumptions.  Below  we write expectations associated to the law
described in Algorithm \ref{alg:USMA}
as $\mathbb{E}[\cdot]$.

\begin{hypA}\label{ass:coup_bridge}
(i) For any $r\in[1,\infty)$ there exist a $C<\infty$ such that for any $\theta\in\Theta$
$$
\int_{\mathbb{R}^d}\|x\|^r\nu_{\theta}(dx) \leq C.
$$
(ii) For any $(r,k)\in[1,\infty)\times\{1,\dots,T\}$ there exist a $C<\infty$ such that for any $(\theta,x)\in\Theta\times\mathbb{R}^d$
$$
\int_{\mathbb{R}^d}\|x'\|^r \bar{f}_{\theta,k-1,k}(dx'|x) \leq C\|x\|^r.
$$
(iii) For any $r\in[1,\infty)$ there exist a $C<\infty$ such that for any $\phi=(\theta,\bar{\theta})\in\Theta^2$ we have
$$
\max\left\{\left(\int_{\mathbb{R}^{2d}}\|x-x'\|^r \check{\nu}_{\phi}(d(x,x'))\right)^{1/r},
\left(\int_{\mathbb{R}^{2d}}\|\nabla_{\theta}\log(\nu_{\theta}(x))-
\nabla_{\bar{\theta}}\log(\nu_{\bar{\theta}}(x')
\|^r \check{\nu}_{\phi}(d(x,x'))\right)^{1/r} 
\right\}
\leq C\|\theta-\bar{\theta}\|.
$$
(iv) For any $(r,k)\in[1,\infty)\times\{1,\dots,T\}$ there exist a $C<\infty$ such that for any $\phi=(\theta,\bar{\theta})\in\Theta^2$ and $(x,\bar{x})\in\mathbb{R}^{2d}$ we have
$$
\left(\int_{\mathbb{R}^{2d}}\|x'-\bar{x}'\|^r \check{f}_{\phi,k-1,k}(d(x',\bar{x})|x,\bar{x})\right)^{1/r} \leq C\left(\|\theta-\bar{\theta}\| +
\|x-\bar{x}\|
\right).
$$
\end{hypA}

\begin{hypA}\label{ass:weight}
For each $k\in\{1,\dots,T\}$ there exist a $0<\underline{C}<\overline{C}<\infty$ such that for any $(\theta,x,x',l,\mathbf{w}_{[k-1,k]}^l)
\in\Theta\times\mathbb{R}^{2d}\times\mathbb{N}_0\times\mathbb{R}^{d\Delta_l^{-1}}$
$$
\underline{C} \leq g_{\theta}(y_k|x')R_{\theta,k-1,k}^l(C_{\theta,k-1,k}^l(x,\mathbf{w}_{[k-1,k]}^l,x')) \leq \overline{C}.
$$
\end{hypA}

\begin{hypA}\label{ass:weight_disc}
(i) For each $(r,k)\in[1,\infty)\times \{1,\dots,T\}$ there exist a $C<\infty$ such that for any $\phi=(\theta,\bar{\theta})
\in\Theta^2$
and $(x,\bar{x},l)
\in\mathbb{R}^{2d}\times\mathbb{N}$
$$
\Bigg(\int_{\mathbb{R}^{d+d\Delta_l^{-1}+d\Delta_{l-1}^{-1}}}
|g_{\theta}(y_k|x')R_{\theta,k-1,k}^l(C_{\theta,k-1,k}^l(x,\mathbf{w}_{[k-1,k]}^l,x'))-
g_{\bar{\theta}}(y_k|\bar{x}')R_{\bar{\theta},k-1,k}^{l-1}(C_{\bar{\theta},k-1,k}^{l-1}(\bar{x},\mathbf{w}_{[k-1,k]}^{l-1},\bar{x}'))|^r\times
$$
$$
\check{f}_{\phi,k-1,k}(d(x',\bar{x}')|x,\bar{x})\check{\mathbb{W}}^l(d(\mathbf{w}_{[k-1,k]}^l,\mathbf{w}_{[k-1,k]}^{l-1}))\bigg)^{1/r} \leq C\left(\Delta_l^{1/2} +
\|\theta-\bar{\theta}\| +
\|x-\bar{x}\| 
\right).
$$
(ii) For each $(r,k)\in[1,\infty)\times \{1,\dots,T\}$ there exist a $C<\infty$ such that for any $\phi=(\theta,\bar{\theta})
\in\Theta^2$
and $(x,\bar{x},l)
\in\mathbb{R}^{2d}\times\mathbb{N}$
$$
\Bigg(\int_{\mathbb{R}^{d+d\Delta_l^{-1}+d\Delta_{l-1}^{-1}}}
\Big\|\nabla_{\theta}\log\left\{g_{\theta}(y_{k}|x')R_{\theta,k-1,k}^{l}\left(
C_{\theta,k-1,k}^l(x,\mathbf{w}_{[k-1,k]}^l,x')\right)\right\} - 
$$
$$
\nabla_{\bar{\theta}}\log\left\{g_{\bar{\theta}}(y_{k}|\bar{x}')R_{\bar{\theta},k-1,k}^{l-1}\left(
C_{\bar{\theta},k-1,k}^{l-1}(\bar{x},\mathbf{w}_{[k-1,k]}^{l-1},\bar{x}')\right)\right\}
\Big\|^r
\check{f}_{\phi,k-1,k}(d(x',\bar{x}')|x,\bar{x})\check{\mathbb{W}}^l(d(\mathbf{w}_{[k-1,k]}^l,\mathbf{w}_{[k-1,k]}^{l-1}))\bigg)^{1/r} \leq 
$$
$$
C\left(\Delta_l^{1/2}+
\|\theta-\bar{\theta}\| +
\|x-\bar{x}\| 
\right).
$$
\end{hypA}

\begin{hypA}\label{ass:grad_bound}
(i)  There exist a $C<\infty$ such that for any $\theta
\in\Theta$ and $x\in\mathbb{R}^{d}$
$$
\|\nabla_{\theta}\log(\nu_{\theta}(x))\|\leq C.
$$
(ii) For each $k\in\{1,\dots,T\}$ there exist a $C<\infty$ such that for any $(\theta, x,\bar{x},l,\mathbf{w}_{[k-1,k]}^l)\in\Theta\times\mathbb{R}^{2d}\times\mathbb{N}_0\times\mathbb{R}^{d(\Delta_l^{-1}-1)}$
$$
\left\|\nabla_{\theta}\log\left\{g_{\theta}(y_{k}|x')R_{\theta,k-1,k}^l\left(
C_{\theta,k-1,k}^l(x,\mathbf{w}_{[k-1,k]}^l,x')\right)\right\}\right\| \leq C.
$$
\end{hypA}

\begin{hypA}\label{ass:mom_cond}
For any $T\in\mathbb{N}$ there exists a $C<\infty$ such that
for any $(l,n)\in\mathbb{N}_0^2$ we have
$$
\mathbb{E}\left[\sum_{k=0}^T\{\|X_k^{n,l}\|+\|\bar{X}_k^{n,l}\|\}\right] \leq C.
$$
\end{hypA}

In the below assumption, for $l\in\mathbb{N}_0$, we denote by $(\theta_n^l)_{n\in\mathbb{N}_0}$ a sequence
of random variables (initial point $\theta^l_0$ as in Algorithm \ref{alg:USMA}) that have been produced by using MSA with transition kernel $K_{\theta,l}$ and initialization $\overline{\mathbb{P}}^l_{\theta^l_0}$.

\begin{hypA}\label{ass:conv_sa}
(i) $\theta_{\star}$ is the unique maximizer of $p_{\theta}(y_1,\dots,y_T)$ with $\nabla_{\theta_{\star}}\log(p_{\theta_{\star}}(y_1,\dots,y_T)) = 0$.\\
(ii) For $l\in\mathbb{N}_0$,  $\theta_{\star}^l$ is the unique maximizer of $p_{\theta}^l(y_1,\dots,y_T)$ with $\nabla_{\theta_{\star}^l}\log(p_{\theta_{\star}^l}^l(y_1,\dots,y_T)) = 0$.  In addition $\lim_{l\rightarrow\infty}\theta_{\star}^l=\theta_{\star}$.\\
(iii) For $l\in\mathbb{N}_0$,  almost surely $\lim_{n\rightarrow\infty}\theta_n^l=\theta_{\star}^l$. 
\end{hypA}

\subsubsection{Discussion of Assumptions}

We discuss each of these assumptions in turn.  (A\ref{ass:coup_bridge}) relates to the coupling of the simulation of the bridge and its continuity in $\theta$ and $x$.  In general this seems a reasonable assumption which says if one generated a path of the bridge with the same $\theta$ at a pair of levels, the same bridge is created. 
(A\ref{ass:weight}) relates to a control of the (forward in time) resampling weight that is used in 
Algorithms \ref{alg:ctcpbs} and \ref{alg:dtccpbs}.  The upper-bound is quite reasonable,  but the lower bound would imply that the diffusion is restricted to a compact space as is the parameter and the observation.  At this time removing the assumption seems somewhat troublesome and we would expect an entirely new approach would be needed. (A\ref{ass:weight_disc}) is associated to the convergence of the time-discretization scheme and continuity. 
The main issue is to discover the time discretization properties of the Euler-Maruyama maps $C_{\theta,k-1,k}^l$.
Although there are some remarks in \cite{papas} in general it seems to be an open problem and a proof of these properties is required.  None-the-less we believe the assumption holds true for regular enough diffusion processes (such as those satisying (D1)). (A\ref{ass:grad_bound}) is strong,  but in fact,  could be substituted by a Markov chain moment condition such as in (A\ref{ass:mom_cond}) and was not to reduce the length of the proofs in the appendix.  (A\ref{ass:mom_cond}) could be established by considering the stability of the Markov chain associated to the MSA method and could follow the lines of the ideas in \cite[Lemma 32]{ub_grad}.  Finally 
(A\ref{ass:conv_sa}) is associated to the convergence of MSA and the uniqueness of solutions; this topic has been
extensively studied in \cite{andr_moulines1,andr,fort}. Indeed \cite{ub_par} consider this question for similar types of CPF and CCPF.

\subsection{Main Result}

We have the following result whose proof is in Appendix \ref{app:main}.

\begin{theorem}\label{theo:main}
Assume (A\ref{ass:coup_bridge}-\ref{ass:conv_sa}).  Then there exist choices of $\mathbb{P}_{\mathtt{L}}$,
$\mathbb{P}_{\mathtt{P}}$,  $(N_p)_{p\in\mathbb{N}_0}$ and $(\gamma_n)_{n\in\mathbb{N}}$ such that
for $j\in\{1,\dots,d_{\theta}\}$:
\begin{eqnarray*}
\mathbb{E}[\widehat{\theta}_{\star}^j] & = & \theta_{\star}^j \\
\mathbb{E}[\{\widehat{\theta}_{\star}^j\}^2] & < & \infty.
\end{eqnarray*}
\end{theorem}

\subsection{Implication of Main Result}

If one looks at the proof of Theorem \ref{theo:main} the suggestion is that $N_p=p+1$,  $\gamma_n=(n+1)^{-(2+\kappa)}$,  $\kappa>0$, $\mathbb{P}_{\mathtt{P}}(p)\propto (p+1)^{-(1+\kappa)}$ and $\mathbb{P}_{\mathtt{L}}(l)\propto\Delta_l^{\tilde{\vartheta}_2}$,  $\tilde{\vartheta}_2>0$.  In general this suggestion is not optimal and relates to our proof strategy as we now detail.  The approach that we use relies extensively on the ideas in \cite{ub_grad} which are associated to one-step properties of  $\check{K}_{\phi,l}$.  This in itself is reasonable, but the approach in \cite{ub_grad} is not sharp enough to maintain the correct coupling rate in $l$; see the discussion of \cite[Section 4]{ub_grad}.  Then, in addition,  we do use the convergence of the MSA method, instead relying naively on the SA recursion formula.  To improve our analysis in the latter direction,  one could use the Martingale plus remainder 
decomposition used (for example) in \cite{fort}.  The complication that this would yield is that the afore-mentioned one-step properties associated to $\check{K}_{\phi,l}$ would have to relate to solutions of the Poisson equation for Markov chains \cite{glynn} which,  given the CPF and CCPF that we use, would be very technically challenging; none-the-less we see this as an interesting avenue for future work.

Now,  considering a sensible setting of the parameters,  we believe that a sharp bound on $\mathbb{E}[\{\widehat{\theta}_{\star}^j\}^2]$ for any $j\in\{1,\dots,d_{\theta}\}$ would be an expression of the type
$$
\sum_{(l,p)\in\mathbb{N}_0}\frac{\Delta_l^{\tfrac{1}{2}}\gamma_{N_p}}{\mathbb{P}_{\mathtt{P}}(p)\mathbb{P}_{\mathtt{L}}(l)}.
$$
Then one could set $N_p=2^p$,  $\gamma_n=(n+1)^{-1}$,  $\mathbb{P}_{\mathtt{P}}(p)\propto 2^{-p}(p+1)\log_2(p+2)^2$,  $\mathbb{P}_{\mathtt{L}}(l)\propto 2^{-l/2}(l+1)\log_2(l+2)^2$ and this would give a finite variance and unbiased estimator.  \cite[Section 3]{ub_filt} provide a discussion of the associated computational complexity and we refer the reader there.

\section{Numerical Results}\label{sec:numerics}

In this section, we consider two implementations of the methodology developed. The first is based on an 
Ornstein-Uhlenbeck (OU) diffusion with Gaussian observations, while the second uses a logistic diffusion with negative binomial observations. The OU model serves as a convenient test case since it allows us to compute various quantities, such as the filtering distribution, the smoothing distribution, the likelihood of the observations, and its gradient, making possible to perform comparisons at every step of the algorithm. The logistic diffusion, in contrast, is applied to real data with irregular observation times (as mentioned earlier, the extension to irregularly spaced observations is theoretically straightforward).

The choices required for the algorithm (aside from hyperparameters such as the number of particles, the step size of the SA, and others) are the auxiliary process $\tilde{X}_t$ and the proposal transition $\bar{f}_{\theta,k-1,k}( x_k|x_{k-1})$. For both implementations, we are able to identify natural choices for these components that closely approximate the original process. The code is written in Python and it is available in \url{https://github.com/maabs/Bridge}.

\subsection{OU Example}

The OU process is defined as  
\begin{align*}
    dX_t = \theta_1 X_t \, dt + \theta_2 \, dW_t, 
    \quad X_0 = x_0, 
    \quad t \in \mathbb{R}^+,
\end{align*}
where $X_t \in \mathbb{R}$, $\theta_1 < 0$, and $\theta_2 > 0$. Observations are made at unit times and follow  
\[
    Y_t \sim \mathcal{N}(X_t, R), \qquad R = \theta_3^2, \ \theta_3 > 0.
\]
In this problem we know not only the exact solution of the OU SDE but we can obtain the likelihood and its gradient using Kalman smoothing.  The gradient field of the log-likelihood function is shown in Figure~\ref{fig:GF}. 
The score function is expressed in terms of synthetic observations from a realization of the system 
with parameters $T=10$, $(\theta_1,\theta_2,\theta_3)=(-0.3,0.8,0.55)$. 
Figure~\ref{fig:GF} is plotted using $\theta_3=0.55$.  
\begin{figure}[h!]
  \centering
  \includegraphics[width=0.8\textwidth,height=6cm]{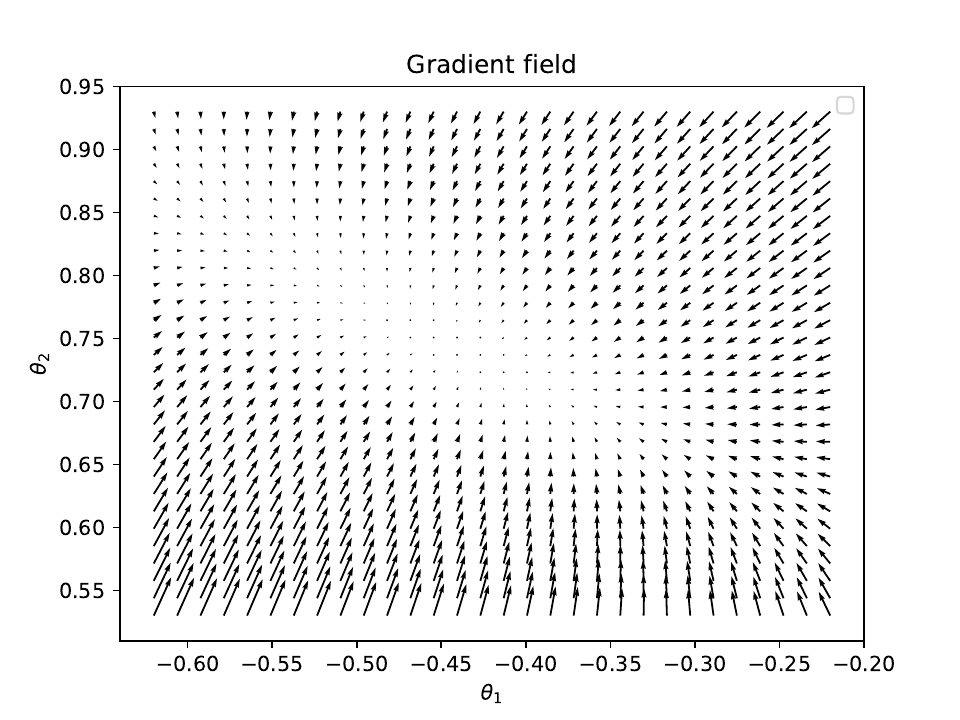}
  \caption{Gradient field of the log-likelihood. The observation of the score function comes from a realization of the observation process with time $T=10$, conditioned on a realization of the diffusion. The parameters of the system are $(\theta_1,\theta_2,\theta_3)=(-0.3,0.8,0.55)$.}
  \label{fig:GF}
\end{figure}

The auxiliary process is taken to be a Gaussian process with different parameters, i.e.,
\[
    d\tilde{X}_t = \tilde{\theta}_1 \tilde{X}_t \, dt + \tilde{\theta}_2 \, dW_t, 
    \quad \tilde{X}_0 = x_0.
\]
Absolute continuity requires $\theta_2 = \tilde{\theta}_2$. The transition probability of this process is 
\[
    \tilde{f}_{\theta,k-1,k}(x_k|x_{k-1}) 
    = \varphi\!\left(x_k;\tilde{F}x_{k-1},\tilde{V}\right),
\]
with $\tilde{F} = e^{\tilde{\theta}_1}$ and 
\[
    \tilde{V} = \frac{\tilde{\theta}_2^{2}}{2\tilde{\theta}_1}\left(e^{2\tilde{\theta}_1} - 1\right),
\]
where $\varphi(x;\mu,\sigma^2)$ denotes the normal density with mean $\mu$ and variance $\sigma^2$.

The proposal bridge from $X_{s_1}^\circ = x$ to $X_{s_2}^\circ = x'$, $0 < s_1 < t < s_2$, is defined by  
\[
    dX_t^\circ 
    = \left[\theta_1 X_t^\circ + \theta_2^2 \nabla_x \log \tilde{f}_{\theta,t,s_2}( x'|X_t^\circ) \right] dt 
      + \theta_2 \, dW_t, 
    \quad X_{s_1}^\circ = x.
\]
The difference between the drifts of the original process and the auxiliary process is 
\[
    \mu_\theta(t,x) - \tilde{\mu}_\theta(t,x) = x(\theta_1 - \tilde{\theta}_1),
\]
and the difference between the squared diffusion coefficients is 
\[
    \Sigma(t,x) - \tilde{\Sigma}(t,x) = 0,
\]
where the latter result spares us from computing $\nabla_x^2 \log \tilde{f}_{\theta,t,s}(x'|x)$. As the proposal transition, we choose the distribution of the auxiliary process.  

The parameters chosen for the original model are $(\theta_1,\theta_2,\theta_3) = (-0.3,0.8,0.8)$, and the auxiliary process has parameter $\tilde{\theta}_1 = -0.1$.
In the following, we discuss the parameters for the implementation of the unbiased estimator and its estimates of error and cost. The choice of the number of particles, $N$, was made by examining the mixing of the MCMC, as illustrated in Figure~\ref{fig:mcmc_steps}. Mixing improves as $N$ increases: for $N=10$, there are several instances in which a sample variable is rejected (at some times, though not necessarily all); this effect is less pronounced for $N=30$, and for $N=50$ it is almost negligible. For our simulations we therefore set $N=50$. We note that backward sampling allows \emph{partial rejection}, meaning that in some paths, rejections occur only at certain times rather than across the entire path.
 
The stochastic approximation step we use is 
\[
    \gamma_n = \gamma_0 \frac{1}{n^{\alpha+0.5}}, 
    \qquad (\gamma_0,\alpha) = (0.2,0.5).
\]
The value we choose through the whole project is $\alpha=0.5$.
We say a time discretization $l\in \mathbb{Z}^+$ has discretization step size $\Delta_l = 2^{-l}$. Similarly, a discretization in the number of steps is parametrized by $p \in \mathbb{Z}^+$, with the sequence of step counts $n_p = n_0 2^p$, where $n_0 \in \mathbb{Z}^+$.  

Given the practical nature of computing the unbiased estimator, it is necessary to truncate the probability distributions of the levels. This truncation prevents unbounded computation times and excessive memory usage. The truncation is specified by the maximum levels of discretization, denoted by $L_{\text{Max}}$ and $P_{\text{Max}}$.  
The probability densities of the levels for the unbiased estimator are defined as
\begin{align}
      \mathbb{P}_L(l) \propto (q+l-l_0)\log^2(q+l-l_0)\,\frac{1}{2^{l}}, 
    \qquad l \in \{l_0+1,\ldots,L_{\text{Max}}\},
\label{eq:P_l}
\end{align}  
with $q=4$, $l_0=4$, and $L_{\text{Max}}=11$. The parameter $q$ is introduced to ensure that the probability function is decreasing in $l$ over its domain. Likewise,
\begin{align}
    \mathbb{P}_P(p) \propto (q+p)\log^2(q+p)\,\frac{1}{2^{p}}, 
    \qquad p \in \{1,\ldots,P_{\text{Max}}\},
    \label{eq:P_p}
\end{align}
 with $P_{\text{Max}}=14$.   The parameters $P_{\text{Max}}$ and $L_{\text{Max}}$ are chosen so that the resulting truncation bias is negligible compared to the target mean square error (MSE). The initial levels $l_0$ and $n_0$ are selected so that the variance of the base levels is comparable to the second moment of the difference in subsequent level couplings.  

The parallel computation of the unbiased estimators was implemented in two layers. The first layer used the \href{https://docs.python.org/3/library/multiprocessing.html}{\texttt{multiprocessing}} library in Python, with a single node of 40 cores and one core allocated per independent computation. The second layer leveraged the high-performance computing cluster \href{https://docs.hpc.kaust.edu.sa/systems/ibex/}{Ibex}, using embarrassingly parallel computations across 32 different nodes.  

Using the analytical score function in conjunction with the gradient descent method, we approximate the analytical solution of the system (the maximum likelihood parameters). This solution $\theta^\star_{\text{app}}$ is then taken as the true parameter values in the approximation of the MSE.  
We approximate the MSE by drawing from a pool of $519 \times 10^3$ independent samples of single unbiased estimators. 
From this pool, we construct $B$ unbiased estimators, each defined as the average of $M$ single estimators $\Xi_{M}^i$, $i\in \{1,\cdots,B\}$, where  $\frac{1}{M}\sum_{j=1}^M \Xi^j $. The MSE is then approximated by applying
\begin{align}
    \varepsilon^2 = \frac{1}{B}\sum_{i=1}^B \left(\Xi_{M}^i - \theta^{\star}_{\text{app}}\right)^2.
    \label{eq:mse_app}
\end{align}

Figure \ref{fig:unb_ou} has been constructed by increasing amount of samples $M$ for the average, increasing the cost and decreasing the MSE this way. We can observe that the error-to-cost rates of our algorithm are slightly smaller than the canonical Monte Carlo error-to-cost rate, as expected from unbiased estimators.

\begin{figure}[htbp]
    \centering
        \includegraphics[width=0.7\linewidth,height=5cm]{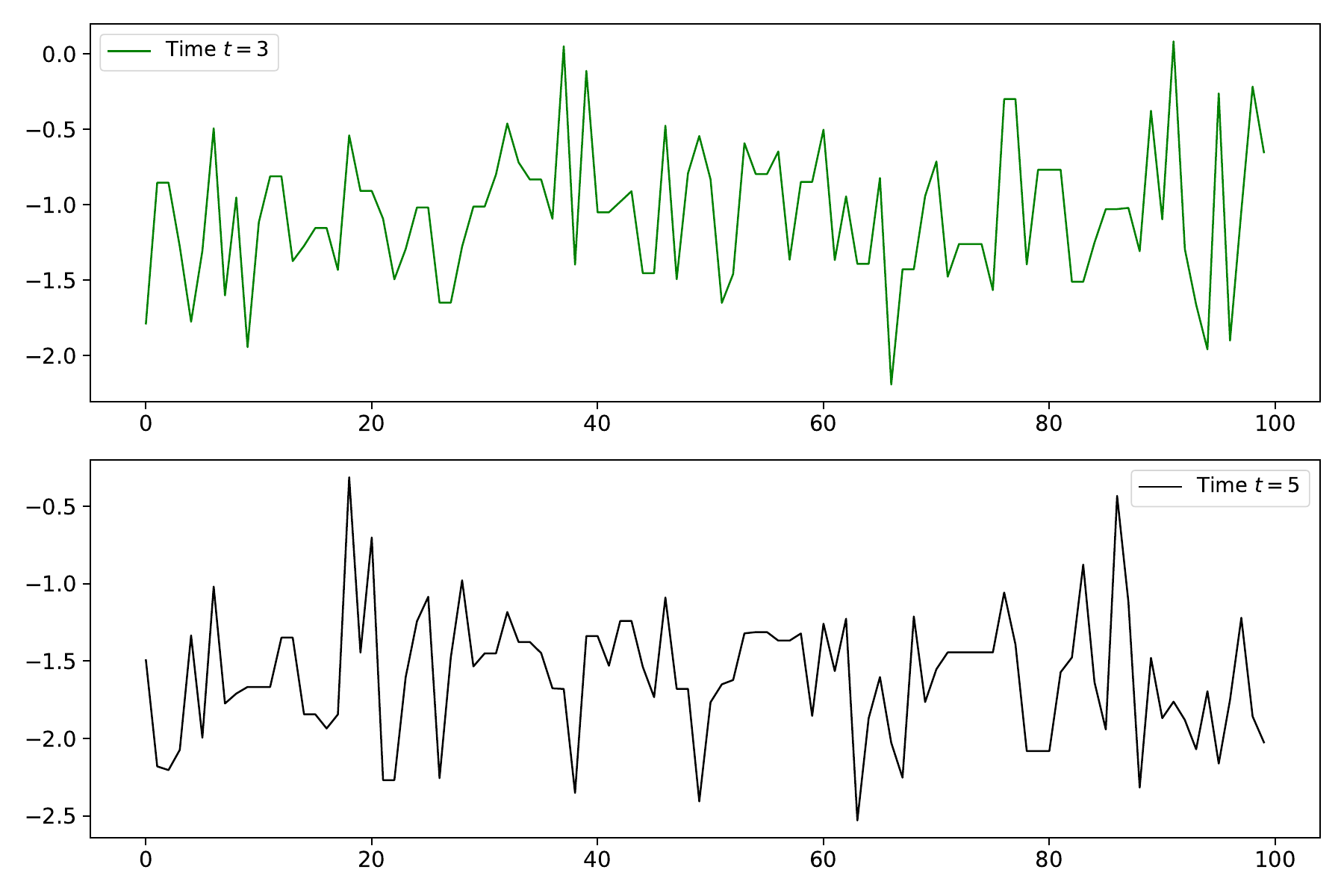}
        \includegraphics[width=0.7\linewidth,height=5cm]{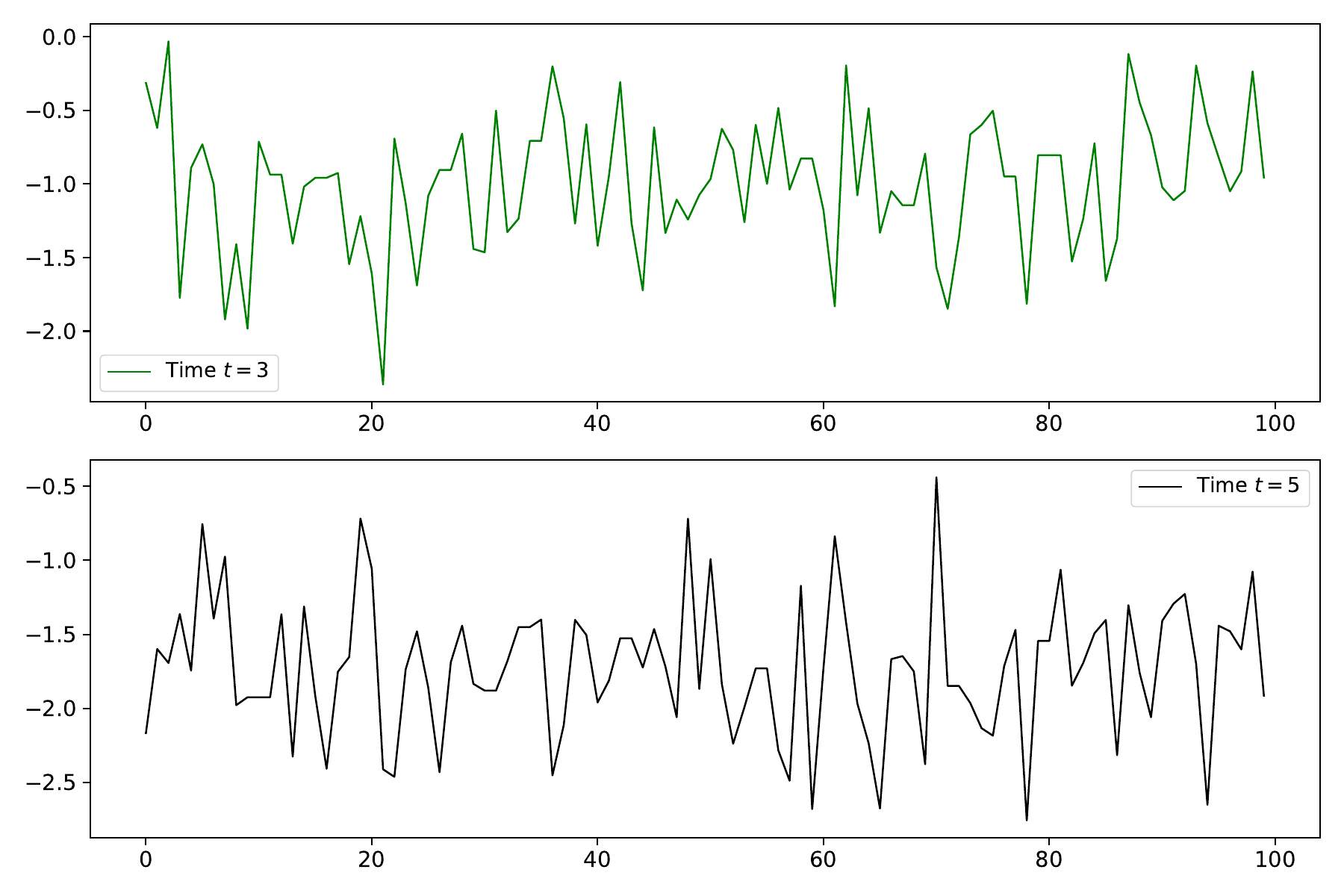}
        \includegraphics[width=0.7\linewidth,height=5cm]{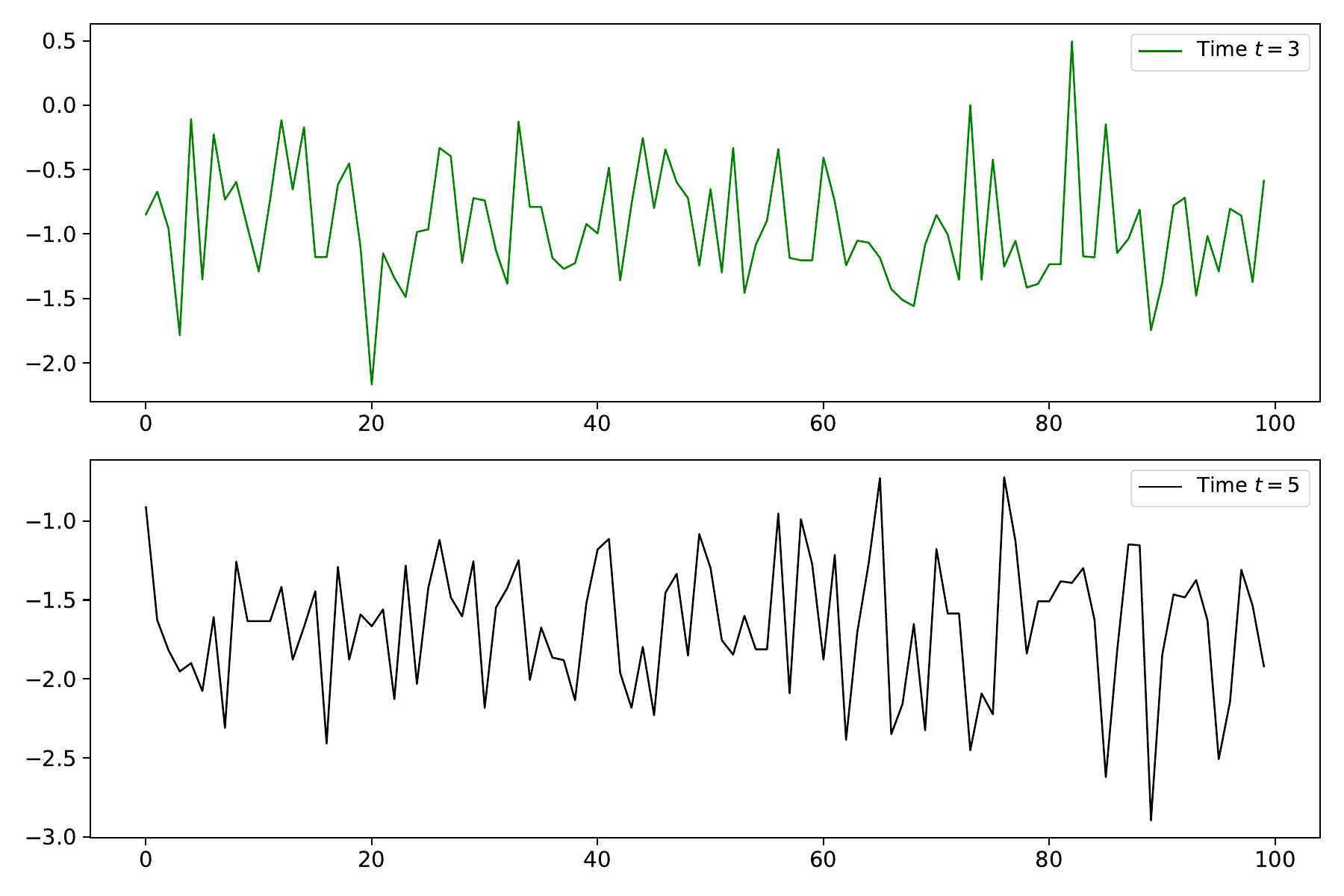}
    \caption{Plots of smoothing samples in terms of the number of MCMC steps.  The plots are in pairs corresponding to $N=10$,  then $N=30$ and finally $N=50$ from top to bottom, each element of the pair corresponds to times $t=3$ and $t=5$.}
        \label{fig:sub1}
    \label{fig:mcmc_steps}
\end{figure}

\begin{figure}[htbp]
  \centering
\includegraphics[width=0.5\linewidth,height=4cm]{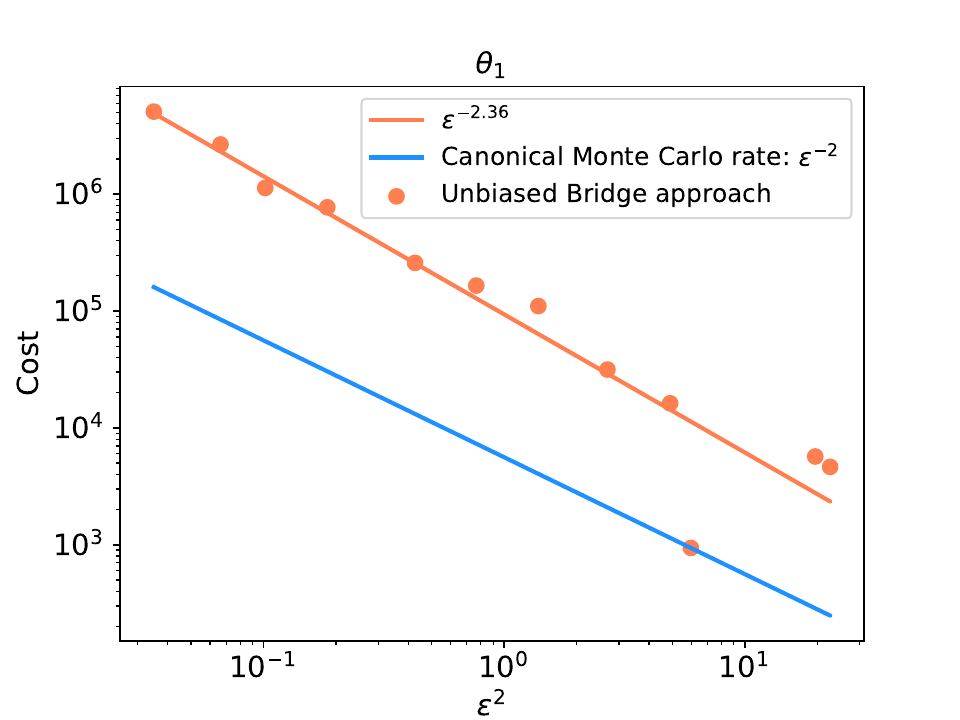}
    \includegraphics[width=0.5\linewidth,height=4cm]{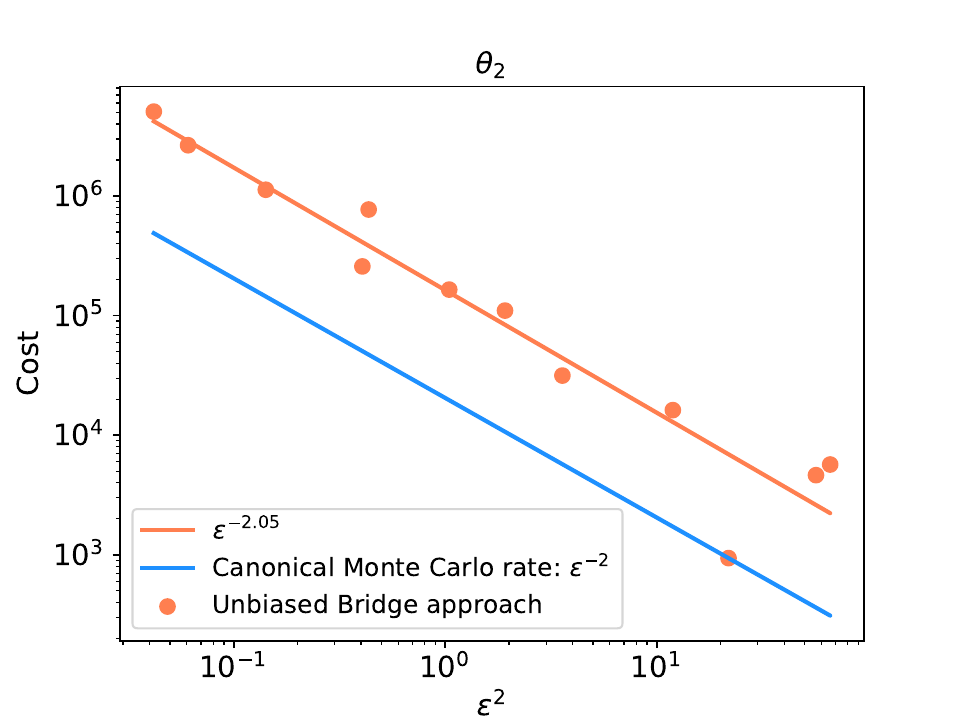}
    \includegraphics[width=0.5\linewidth,height=4cm]{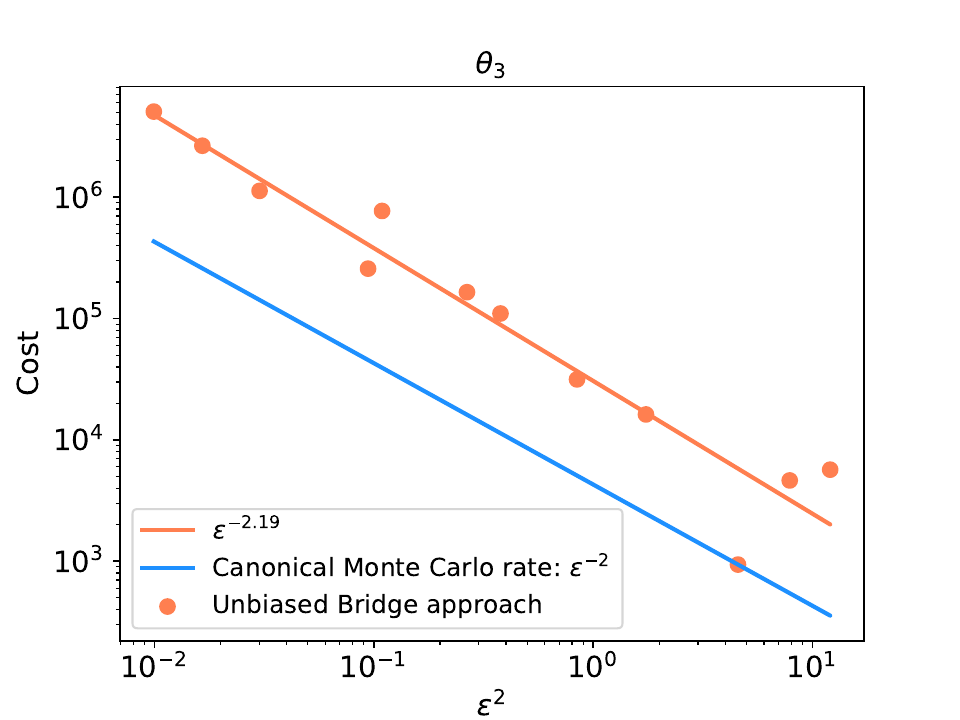}
  \caption{Computation of the estimated MSE ($\varepsilon^2$) vs the cost of computing the Unbiased estimator for the parameters $(\theta_1,\theta_2,\theta_3)$.}
  \label{fig:unb_ou}
\end{figure}

\subsection{Logistic Diffusion}

We consider a case from population ecology, focusing on forecasting the population trends of red kangaroos (Macropus rufus) in New South Wales, Australia \cite{cs87}. The unobserved population size  $X = \{X_t\}_{t \ge t_1}$ (with $t_1 > 0$) is modeled using a logistic diffusion framework \cite{dc88} that incorporates environmental variability. 

\begin{align}
dX_t=\left(\frac{\theta_3^2}{2}+\theta_1-\theta_2 X_t\right)X_tdt +\theta_3 X_t dW_t, \quad X_{t_1}\sim \Gamma (\alpha,\lambda)
\label{eq:kangSDE}
\end{align}
for $\alpha = \tfrac{\theta_1}{2\theta_3^2}$ and $\lambda = \tfrac{\theta_2}{2\theta_3^2}$ where $\Gamma(\alpha,\lambda)$ represents the Gamma distribution with shape $\alpha$ and rate $\lambda$. We model the initial distribution as a gamma density, which is the invariant distribution of the logistic diffusion. This choice ensures that the process begins in its stationary regime, thereby avoiding transient effects and aligning with the long-run equilibrium that the model is designed to capture \cite{dc88}.  

We consider two sequences of observations  which, conditional on the underlying population, are independent of each other. The observations $(Y_{t_1},Y_{t_2},\cdots,Y_{t_k})$, $t_1<t_2<\cdots<{t_k}$, $Y_{t_i}\in(\mathbb{R}^{+})^2$, are modeled using negative binomial distributions to account for overdispersion \cite{bf53}
\begin{align*}
    Y_{t_i}^s| X_{t_i} \stackrel{\text{ind}}{\sim} \mathcal{NB}\left(\theta_4,\theta_4(x_{t_i}+\theta_4)^{-1}\right), \quad s\in\{1,2\},  i\in \{1,\cdots,k\}, \theta_4\in \mathbb{R}^+.
\end{align*}
$Y_{t_i}=(Y_{t_i}^1,Y_{t_i}^2)$,  $\text{ind}$ means independently,  and the negative binomial density  
\begin{align*}
 \frac{\Gamma(y+\theta_4)}{\Gamma(\theta_4)\,\Gamma(y+1)}
\left(\frac{\theta_4}{X_{t_i}+\theta_4}\right)^{\!\theta_4}
\left(\frac{X_{t_i}}{X_{t_i}+\theta_4}\right)^{\!y},
\qquad y\in \mathbb{Z}^+\cup \{0\}
\end{align*}

In the following we discuss the choices made in terms of auxiliary process and proposal transition.
Following \cite{schauer} we approximate \eqref{eq:kangSDE} as a linear SDE in terms of the population, the most general auxiliary process we can propose is
\begin{align*}
    d\tilde{X}_t=\left[v(t)\tilde{X}_t+c(t)\right]dt+\left[f(t)\tilde{X}_t+r(t)\right] dW_t, 
\end{align*}
the choice $f(t)=\theta_3,$ $r(t)=0$ is straightforward since \eqref{eq:kangSDE} has a linear diffusion coefficient already, this choice satisfies $\Sigma(t_i)=\tilde{\Sigma}(t_i)$, for $i\in \{2,\cdots,k\}$.

For the approximation of the drift coefficient, we decide to choose a family of functions parametrized by the endpoints $(x_{s_1},x_{s_2})$ of the time interval $s_1\leq t\leq s_2$. The auxiliary drift $\tilde{b}_{s_1,x_{s_1};s_2,x_{s_2}}(t,x)=v(t)x+c(t) $ interpolates linearly (in terms of time) the drift $b(x)$ at the endpoints. i.e., 
\begin{align*}
\tilde{b}_{s_1,x_{s_1};s_2,x_{s_2}}(s_1,x_{s_1})&={b}(x_{s_1}),\\
\tilde{b}_{s_1,x_{s_1};s_2,x_{s_2}}(s_2,x_{s_2})&={b}(x_{s_2}).
\end{align*}
Let $v(t)=\zeta t+\eta$, then
\begin{align*}
    \zeta&=\frac{b(x_{s_1})}{x_{s_1}(s_1-s_2)}-\frac{b(x_{s_2})}{x_{s_2}(s_1-s_2)},\\
    \eta&=\frac{b(x_{s_1})s_2}{x_{s_1}(s_2-s_1)}-\frac{b(x_{s_2})s_1}{x_{s_2}(s_1-s_2)},\\
\end{align*}
where the dependence of $\zeta$ and $\eta$ on the endpoints is implicit. Furthermore, $c(t)=0$.

The auxiliary transition is a lognormal density
\begin{align*}
    f_{\theta; s_{1},x_{s_1};s_{2},x_{s_2}}(t,x;s_2,x_{s_2})=\frac{\text{exp}\left(\frac{-(\mu_{s_1,x_{s_1};s_2,x_{s_2}}(t,x)-\log(x_{s_2}))^2}{2\theta_3^2(s_2-t)}\right)}{\sqrt{2\pi\theta_3^2(s_2-t})x_{s_2}},
\end{align*}
where $\mu_{s_1,x_{s_1};s_2,x_{s_2}}(t,x)=\log(x)+\frac{\zeta s_2^2}{2}+(\eta-\frac{\theta_3^2}{2})s_2-\left[\frac{\zeta t^2}{2}+(\eta-\frac{\theta_3^2}{2})t \right]$,
with log gradient and hessian matrix
\begin{align*}
   \frac{\partial}{\partial x} \log f_{\theta; s_{1},x_{s_1};s_{2},x_{s_2}}(t,x;s_2,x_{s_2})=&\frac{-(\mu_{s_1,x_{s_1};s_2,x_{s_2}}(t,x)-\log(x_{s_2}))}{x\theta_3^2(s_2-t)},\\
   \frac{\partial^2}{\partial x^2} \log f_{\theta; s_{1},x_{s_1};s_{2},x_{s_2}}(t,x;s_2,x_{s_2})=&\frac{1-\mu_{s_1,x_{s_1};s_2,x_{s_2}}(t,x)+\log(x_{s_2})}{x^2\theta_3^2(s_2-t)}.
\end{align*}
The notation for the transition probability changes slightly with respect to the previous section to emphasize its dependence on the endpoints. The choice of parameterizing the auxiliary processes on the endpoints is intended to improve the variance of the intrinsic importance sampling method induced by the change of measure of the bridges. 

Improving the importance sampling is used a a lead to choose the proposal transition distribution. The most direct choice for the proposal transition is the auxiliar transition (as in the previous example, OU); unfortunately, it cannot be used given the dependence of the auxiliar transition. Unlike the vanilla particle filter of diffusions, the structure of the particle bridge allows implementation of the auxiliary particle filter \cite{ps99,ps01}, since this is out of the scope of the current project, we leave it for future work.

As the proposal transition, we choose to modify the logistic process by replacing one instance of the variables with the expected value of the invariant distribution, i.e.,
\begin{align*}
    dX_t &= \left(\frac{\theta_3^2}{2} + \theta_1 - \theta_2 \, \mathbb{E}(X_\infty)\right) X_t \, dt 
           + \theta_3 X_t \, dW_t, 
           \quad X_\infty \sim \Gamma(\alpha,\lambda), 
           \quad X_{s_1} = x_{s_1}, \\
    dX_t &= \frac{\theta_3^2}{2} X_t \, dt + \theta_3 X_t \, dW_t,
\end{align*}
which corresponds to a geometric Brownian motion.  
This choice could be improved by replacing the expectation with respect to the invariant distribution by the mean of the particle filter at time $s_1$. We leave this refinement for possible future work.

In the following, we present two classes of results. The first concerns the unbiased estimator, while the second provides a comparison between our approach and that in \cite{ub_par} highlighting the difference between implementing the bridge approach with backward sampling and applying Girasnov's theorem (used in \cite{ub_par}). 
Similar to the previous section, we approximate the MSE by drawing from a pool of $117\times 10^3$ independent samples of single unbiased estimators and then using \eqref{eq:mse_app}.
We choose the parameters  
\[
(\theta_1,\theta_2,\theta_3,\theta_4) \;=\; \left(2,\tfrac{2}{522.8},0.8,10\right),
\]  
the value of \(\theta_2\) is set such that the mean of the invariant density, given by \(\theta_1/\theta_2\), coincides with the empirical mean of the observations, namely \(522.8\).  

The number of particles is fixed at \(N=50\), and the stochastic approximation step size is chosen as $\gamma_0 = 5 \times 10^{-3}\,(2,3,0.6,6)$. The probability of the levels for the SA steps follows~\eqref{eq:P_p}, with parameters \(n_0=1\) and \(P_{\text{Max}}=11\). The probability distribution for the time discretization is slightly modified w.r.t. to the OU example to account for the effectiveness of the coupling:  
\begin{align*}
      \mathbb{P}_L(l) \;\propto\; (q+l-l_0)\,\log^2(q+l-l_0)\,\frac{1}{2^{l/2}}, 
      \qquad l \in \{l_0+1,\ldots,L_{\text{Max}}\},
\end{align*}  
where \(l_0=3\) and \(L_{\text{Max}}=8\).  

The true value of the parameters is approximated  $\theta^\star_{\text{app}}$ by running single estimators with a large number of SA steps and a fine time discretization level ($l=13$). 
In addition, we average the final points of the SA and then rerun the algorithm using the previous result as the initial guess; this procedure is repeated three times.  
From these results (Figure \ref{fig:unb_k}), we draw conclusions similar to those in the previous example: the error-to-cost rates are slightly smaller than the canonical Monte Carlo rates.  

\begin{figure}[htbp]
  \centering
    \includegraphics[width=0.45\linewidth,height=4cm]{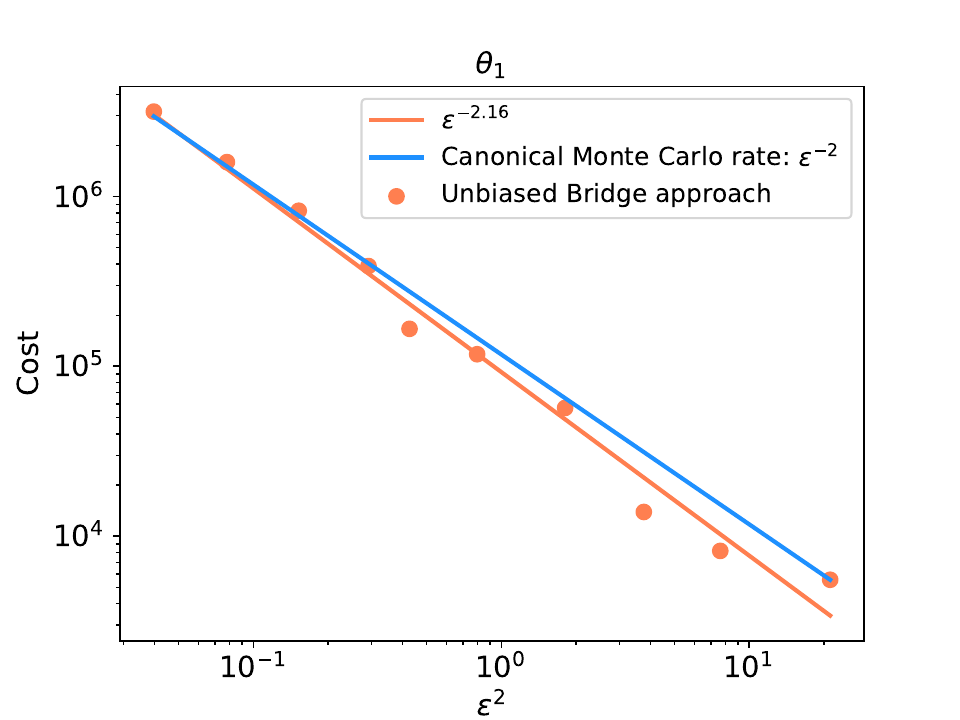}
    \includegraphics[width=0.45\linewidth,height=4cm]{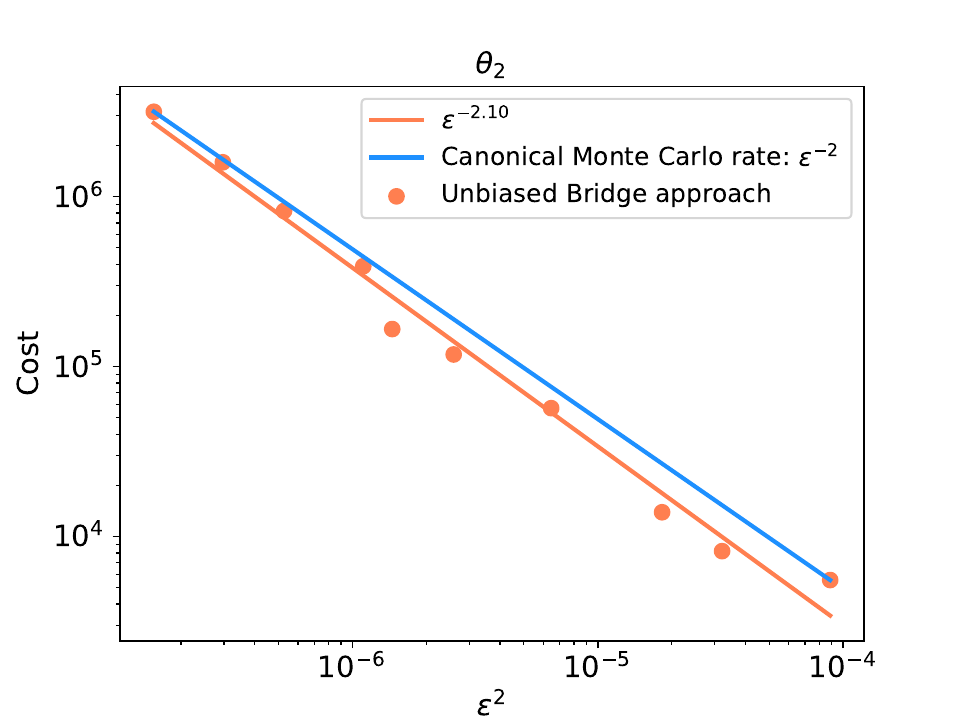}
    \includegraphics[width=0.45\linewidth,height=4cm]{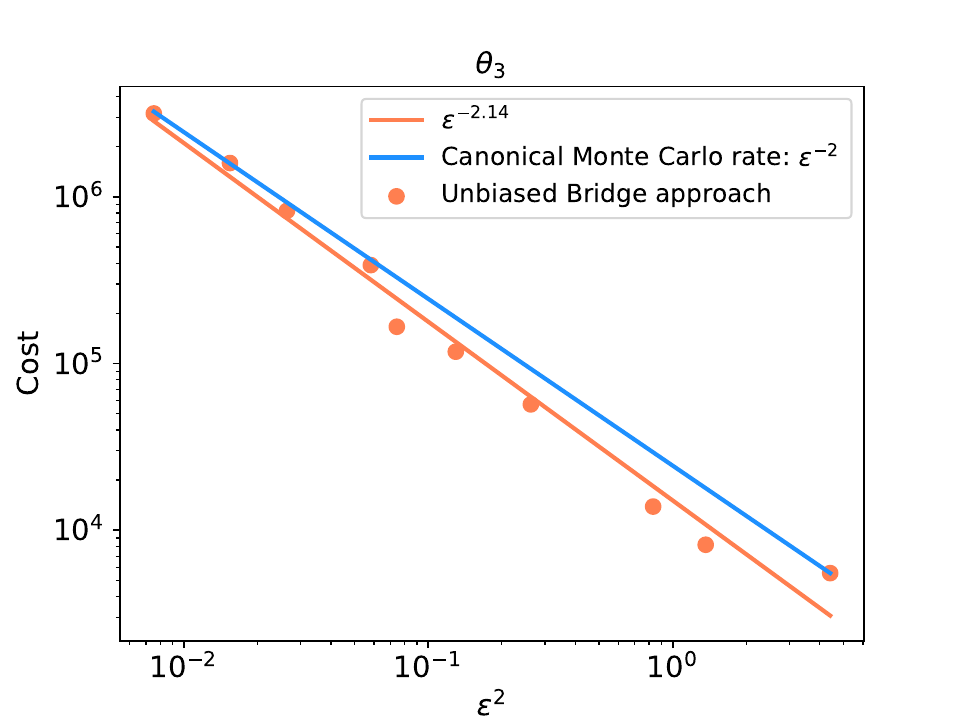}
\includegraphics[width=0.45\linewidth,height=4cm]{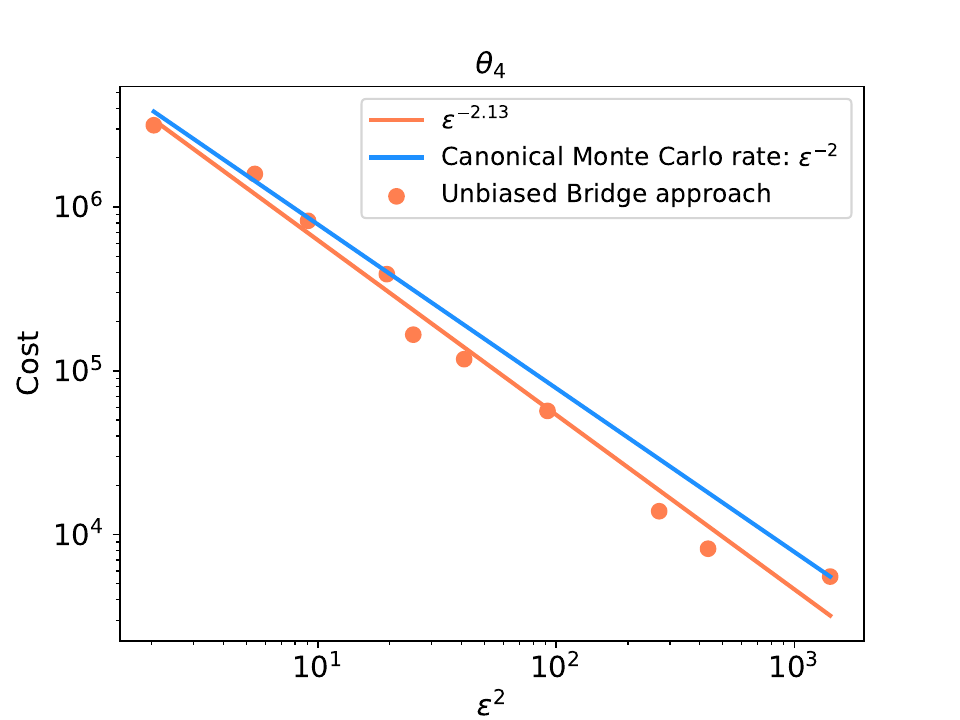}
\caption{Computation of the estimated MSE ($\varepsilon^2$) vs the cost of computing the unbiased estimator for the parameters $(\theta_1,\theta_2,\theta_3,\theta_4)$.}
  \label{fig:unb_k}
\end{figure}

Another way of computing estimators of the score function can be found in \cite{ub_par}. Here, the authors exploit Girsanov's theorem to make a change of variables that allow the computation of the gradient of the likelihood, the change of variable requires the diffusion coefficient not to depend on the static parameters. For comparison of this method to our method, we fix the parameters that depend on the diffusion.

We choose to make this comparison for the following reasons:
	First, the flexibility to allow dependence in the diffusion term comes at the cost of a potential increase in variance, induced by the change of measure associated with the bridges. This increase directly affects the variance of the SA. However, this drawback can be mitigated by carefully selecting the auxiliary process (as in this example, where we use endpoint-dependent auxiliary processes), and it can potentially be further improved by incorporating ideas from the auxiliary particle filter.    
	Second,	our method employs backward sampling within the conditional particle filter, a technique that is naturally supported by diffusion bridges. In the non-bridge setting, backward sampling is still possible but leads to divergences as the time discretization step approaches zero \cite{beskos}. The use of backward sampling mitigates the variance increase induced by the bridge change of measure, improves the mixing rates of the MCMC, and accelerates the convergence of the chain.
    
For the comparison we choose the parameters $(\theta_1,\theta_2,\theta_3,\theta_4)=(2.397, 4.429\times 10^{-3}, 0.84, 17.631)$, with a relatively low number of particles $N=10$, the time step constant of the SA is $\gamma_0=5\times 10^{-4}(2,3,10)$ corresponding to the variables $(\theta_1,\theta_2,\theta_4)$.

The time discretization used in the Girsanov-based method differs in three main ways. First, in our approach the discretization is aligned with the observation times, which results in slightly varying step lengths. In contrast, the approach in \cite{ub_par} employs a uniform step length, and the observation times are approximated to fit this discretization. This discrepancy diminishes as the levels increase. Second, the interpretation of the base level $l_0$ differs. There is an approximate offset of 3.5 levels: let $l_{0,\text{Back}}$ denote the base level of the backward approach and $l_{0,\text{Gir}}$ that of the Girsanov approximation. Then,
$2^{-l_{0,\text{Back}}} \approx 2^{-(l_{0,\text{Gir}}-3.5)}$.
 In practice, we set $l_{0,\text{Back}} = l_{0,\text{Gir}} - 4$. Third, there is a difference arising from the use of diffusion bridges versus the diffusion process itself. Strictly speaking, this is not a difference in the discretization of the time grid (as in the two previous cases), but it affects the intrinsic discretization within the particle filters, thereby making the two methods distinct.

The quantities we compare are the approximation of the score function and the unbiased estimator for a fixed time discretization and SA discretization. To approximate the true values of both (and thereby estimate the error), we use a large number of particles, $N=500$ to ensure that the chain converges relatively quickly, and we sample a large number of chain iterations (1000) to obtain a reliable representation of the score and the unbiased estimator. In addition, we generate a large number of independent samples, $M=500$, to further reduce the statistical error of our approximation of the true values. This procedure is applied to both our algorithm and the Girsanov-based method.

To ensure a meaningful and fair comparison, we take the final estimator of the truth to be the average of the quantities obtained with both methods. These values differ only by a small percentage, which we attribute to the different discretizations used in the two approaches, as explained earlier.
Let $H_5^i$, $i \in \{0,1,\ldots,1200\}$, denote the i.i.d. samples of the score function, and let ${H}_{5,\text{app}}^\star$ be our approximation of the true value at level $l=5$ ($l=9$ for the Girsanov approach). We define the error as $\varepsilon_i = H_5^i - {H}_{5,\text{app}}^\star$. Figure~\ref{fig:bp_score} displays the corresponding boxplots for both methods, for estimating the score function. We observe that the bridge method produces a distribution centered around zero with smaller variance and substantially fewer outliers.  Similarly, let $\Xi_5^i$, $i \in \{0,1,\ldots,1200\}$, denote the unbiased estimator at level $l=5$ with $p=1$, where $n_0=1$ (the initial number of SA approximations is one). The error for the unbiased (parameter) estimator is defined analogously to the score function error. Figure~\ref{fig:bp_unbiassed} displays the corresponding error distribution boxplots, from which conclusions similar to those drawn for the score function can be deduced.

\begin{figure}  
  \centering
    \includegraphics[width=0.45\linewidth,height=4cm]{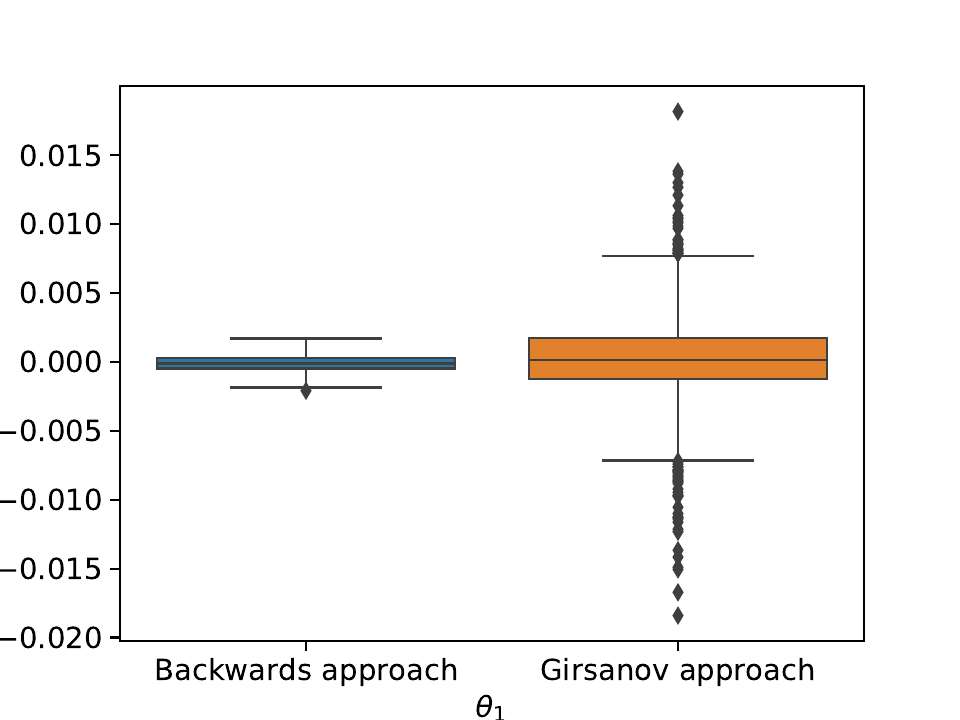}
    \includegraphics[width=0.45\linewidth,height=4cm]{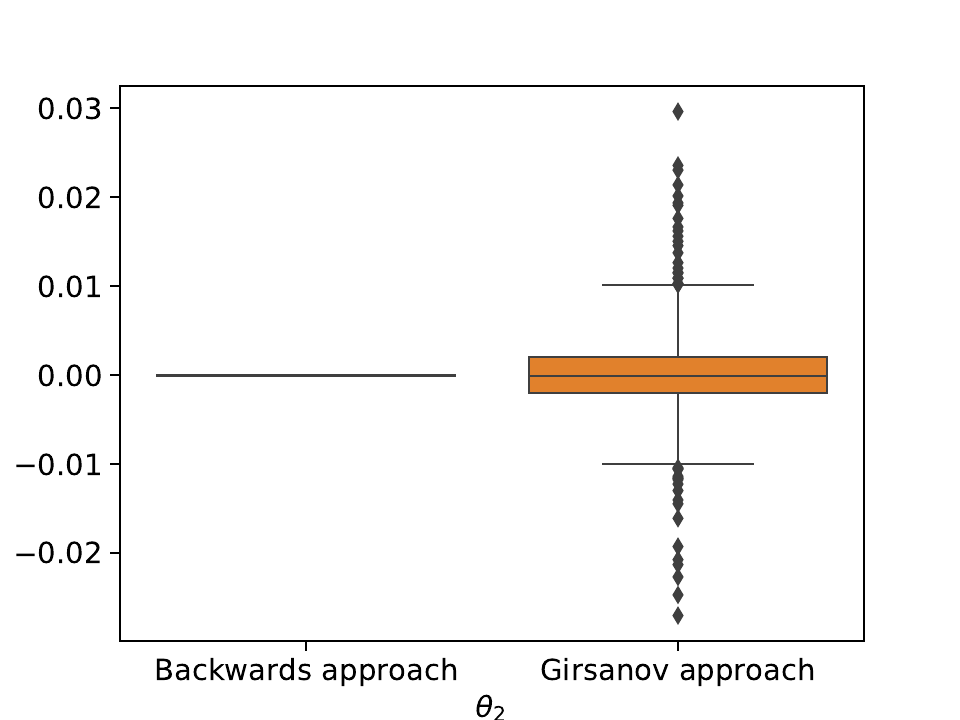}
    \includegraphics[width=0.45\linewidth,height=4cm]{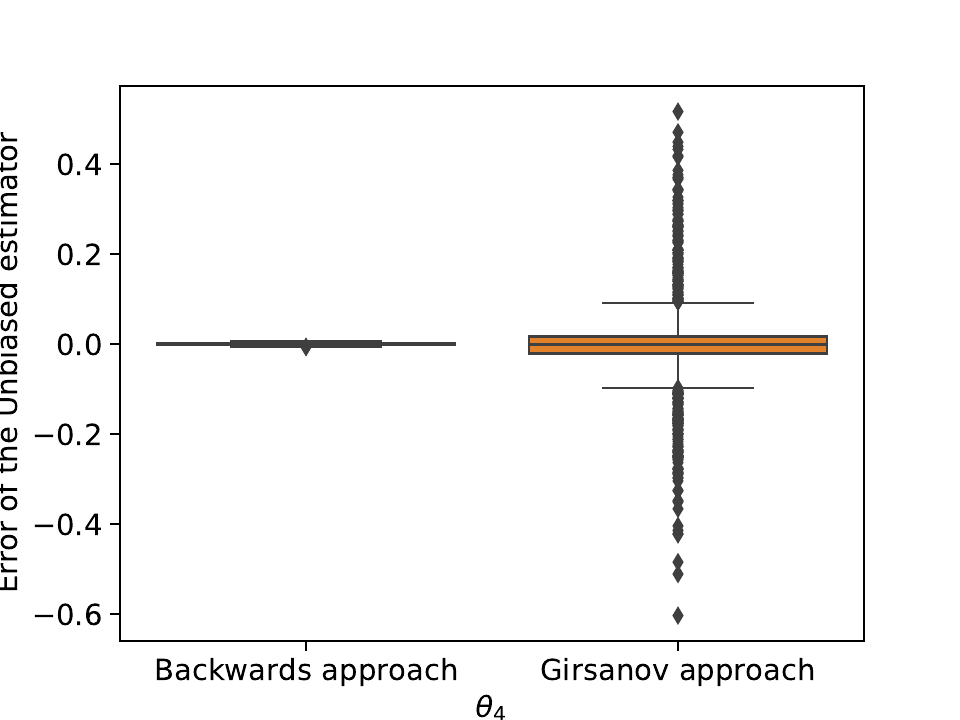 }
  \caption{Boxplots of the errors of the estimators of the score function. Each plot contains both the bridge/backward approach and the Girsanov approach for either $\theta_1,$ $\theta_2$ or $\theta_4$.}
\label{fig:bp_score}
\end{figure}

\begin{figure}
  \centering
    \includegraphics[width=0.45\linewidth,height=4cm]{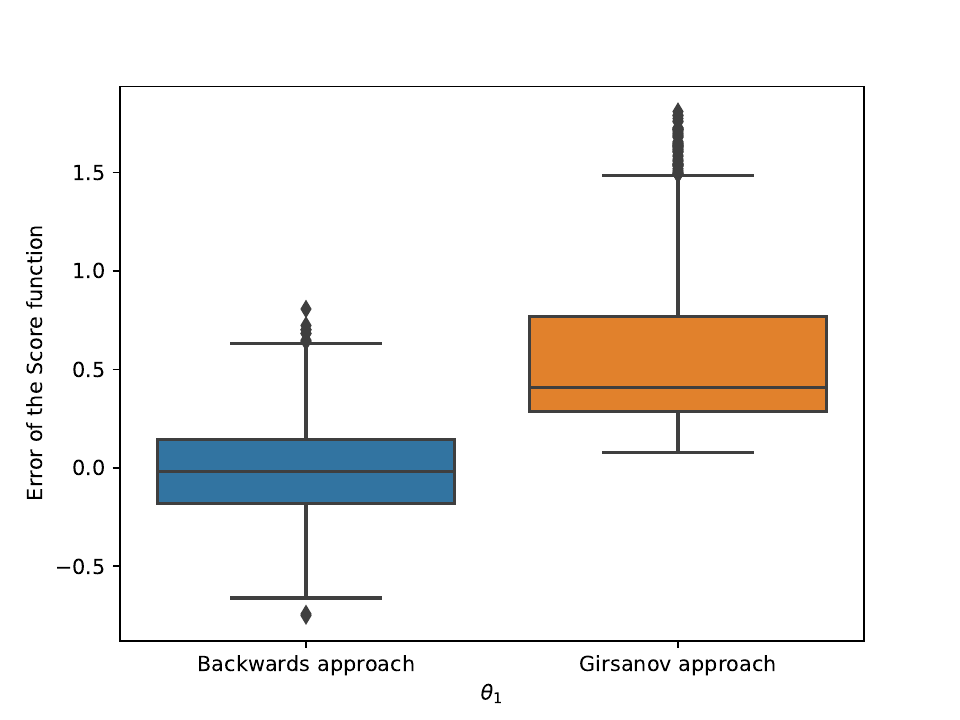}
    \includegraphics[width=0.45\linewidth,height=4cm]{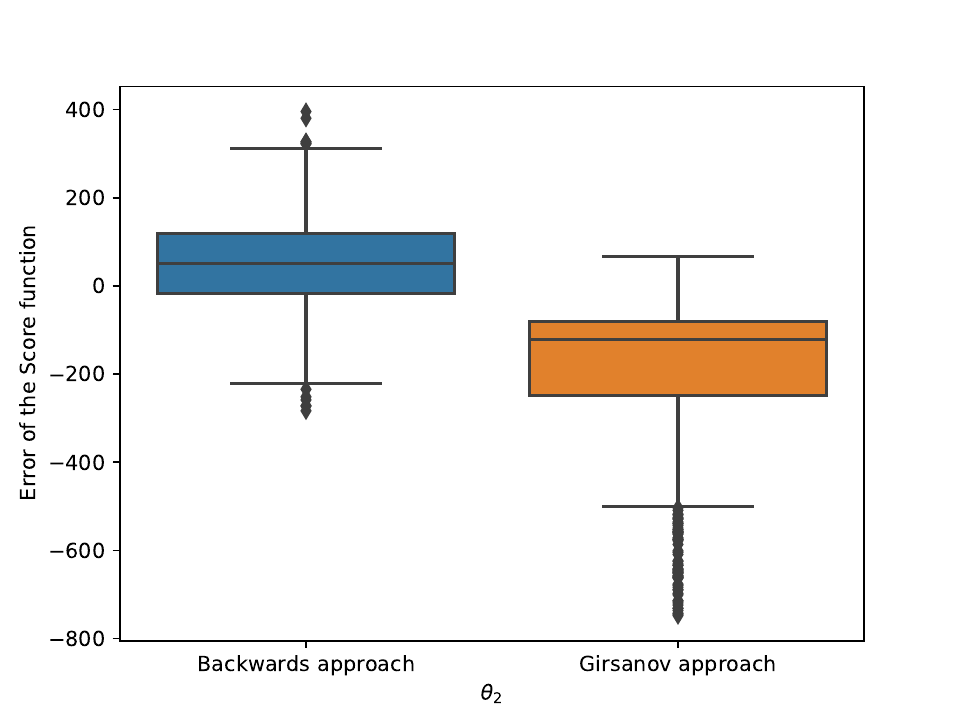}
    \includegraphics[width=0.45\linewidth,height=4cm]{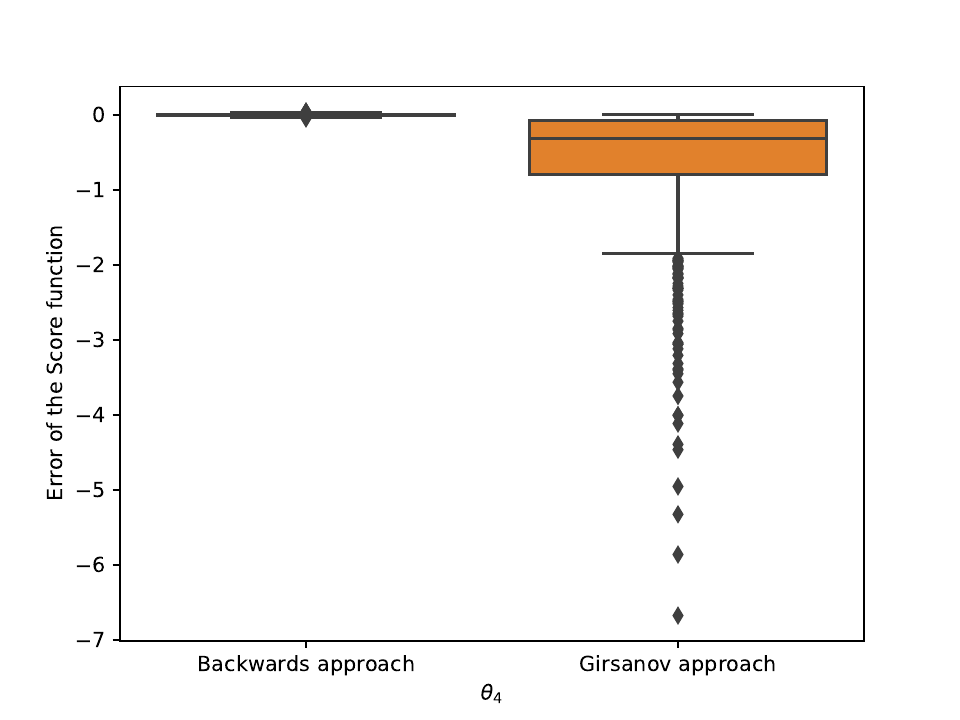 }
  \caption{Boxplots of the errors of unbiased estimators. Each plot contains both the bridge/backward approach and the Girsanov approach for either $\theta_1,$ $\theta_2$ or $\theta_4$.  }
\label{fig:bp_unbiassed}
\end{figure}

\appendix

\section{Mathematical Results}\label{app:math}

\subsection{Structure}

The purpose of this appendix to provide a proof of Theorem \ref{theo:main}.  In order to do this several technical results regarding several elements of Algorithm \ref{alg:USMA} have to be established.  The structure of the appendix,  which should be read in order to understand its contents,  is as follows.  In Appendix \ref{app:1}
we give some results for Algorithm \ref{alg:dtccpbs} when only considering Step 1.~to 3.  In Appendix 
\ref{app:2} further technical results associated to the entire of Algorithm \ref{alg:dtccpbs} (with \textbf{(FC)} replacing Step 5.) are provided.  In Appendix \ref{app:3} we analyze an initialization in Algorithm \ref{alg:USMA} and in Appendix \ref{app:4} Algorithm \ref{alg:USMA}  itself.  Finally in Appendix \ref{app:main}
the proof of Theorem \ref{theo:main} is given.

\subsection{Results for Algorithm \ref{alg:dtccpbs} from Step 1.~to 3.}\label{app:1}

For any $(i,k,s)\in\{1,\dots,N\}\times\{1,\dots,T\}\times\{l-1,l\}$, 
we will write $a_k^{s}(i) = a_k^{i,s}$ and $\bar{a}_k^{s}(i) = \bar{a}_k^{i,s}$. 
Using this notation, for any $(k,l)\in\{1,\dots,T-1\}\times\mathbb{N}$, we define
\begin{align*}
& \mathsf{S}_k^l=\{i\in\{1,\ldots, N\}: a_{k}^l(i)=\bar{a}_{k}^{l-1}(i),a_{k-1}^l\circ a_{k}^l(i)=\bar{a}_{k-1}^{l-1}\circ \bar{a}_{k}^{l-1}(i),\dots,
	a_{1}^l\circ\cdots\circ a_{k}^l(i)=\bar{a}_{1}^{l-1}\circ\cdots\circ \bar{a}_{k}^{l-1}(i)\}.	
\end{align*}
We will set $\mathsf{S}_0^l=\mathsf{S}_{-1}^l=\{1,\dots,N\}$.  In addition we use the convention that
$a_{0}^{i,l}=\bar{a}_{0}^{i,l-1}=i$.
We define for $t\in\{0,\dots,T\}$
$$
H_{l,t}(\theta,\mathbf{z}_t^l) := \sum_{k=1}^t \nabla_{\theta}\log\left\{g_{\theta}(y_{k}|x_{k})R_{\theta,k-1,k}^l\left(
C_{\theta,k-1,k}^l(x_{k-1},\mathbf{w}_{[k-1,k]}^l,x_k)\right)\right\}+\nabla_{\theta}\log(\nu_{\theta}(x_0))
$$
where we use the short-hand $\mathbf{z}_t^l = (z_1^l,\dots,z_t^l)$.  We set $\mathsf{Z}^l:=\mathbb{R}^{d(T\Delta^{-1}_l+1)}$.

For $(l,\vartheta,C)\in\mathbb{N}\times\mathbb{R}^+\times\mathbb{R}^+$ and $(\theta,\bar{\theta})\in\Theta^2$ we introduce the following sets,  where we use the notation $x_{0:T}=(x_0,\dots,x_T)$
\begin{align*}
\mathsf{B}^l_{\vartheta,C}(\theta,\bar{\theta}) & = \{
(\mathbf{z}^l,\bar{\mathbf{z}}^{l-1})\in\mathsf{Z}^l\times\mathsf{Z}^{l-1}:
|g_{\theta}(y_k|x_k)R_{\theta,k-1,k}^l(C_{\theta,k-1,k}^l(x_{k-1},\mathbf{w}_{[k-1,k]}^l,x_k))-\\
&
g_{\bar{\theta}}(y_k|\bar{x}_k)R_{\bar{\theta},k-1,k}^{l-1}(C_{\bar{\theta},k-1,k}^{l-1}(\bar{x}_{k-1},\mathbf{w}_{[k-1,k]}^{l-1},\bar{x}_k))| \leq C(\Delta_l^{\vartheta}+\|\theta-\bar{\theta}\|),  k\in\{1,\dots,T\}
, \\ &
\|x_k-\bar{x}_k\| \leq C\Delta_l^{\vartheta}, k\in\{0,\dots,T\}\},\\
\mathsf{H}^l_{\gamma,C}(\theta,\vartheta) & = 
\{(\mathbf{z}^l,\bar{\mathbf{z}}^{l-1})\in\mathsf{Z}^l\times\mathsf{Z}^{l-1}:
\|H_{l,k}^l(\theta,\mathbf{z}_k^l) - H_{l-1,k}^l(\bar{\theta},\bar{\mathbf{z}}_k^{l-1})\| \leq C(\Delta_l^{\vartheta}+\|\theta-\bar{\theta}\|),k\in\{0,\dots,T\}\}.
\end{align*}

The results in this section regard Algorithm \ref{alg:dtccpbs} from Step 1.~to 3.~only.   Expectations associated to the transition kernel in Algorithm \ref{alg:dtccpbs} from Step 1.~to 3. ~are denoted as $\bar{\mathbb{E}}_{\phi,l}[\cdot]$, where $\phi=(\theta,\bar{\theta})$.  We will use the notation $\mathbf{z}_1^{i,l}=z_1^{i,l}=(x_0^{i,l},\mathbf{w}_{[0,1]}^{i,l},x_1^{i,l})$,  $\bar{\mathbf{z}}_1^{i,l-1}=\bar{z}_1^{i,l-1}=(\bar{x}_0^{i,l-1},\bar{\mathbf{w}}_{[0,1]}^{i,l-1},\bar{x}_1^{i,l-1})$  
and for $k\in\{2,\dots,T\}$:
\begin{eqnarray*}
\mathbf{z}_k^{i,l} & =&  (\mathbf{z}_{k-1}^{a_{k-1}^{i,l},l},z_k^{i,l}) \\
\bar{\mathbf{z}}_k^{i,l-1} & =&  (\bar{\mathbf{z}}_{k-1}^{\bar{a}_{k-1}^{i,l-1},l-1},\bar{z}_k^{i,l-1})
\end{eqnarray*}
where $z_k^{i,l}=(\mathbf{w}_{[k-1,k]}^{i,l},x_k^{i,l})$,  $\bar{z}_k^{i,l-1}=(\bar{\mathbf{w}}_{[k-1,k]}^{i,l-1},
\bar{x}_k^{i,l-1})$.

\begin{lem}\label{lem:lem1}
Assume (A\ref{ass:coup_bridge}-\ref{ass:weight}).  Then  for any $(k,r,C')\in\{0,\dots,T\}\times[1,\infty)\times\mathbb{R}^+$ there exists a $C<\infty$ such that for any $(l,\vartheta,N,\phi)\in\mathbb{N}\times\mathbb{R}^+\times\{2,3,\dots\}\times\Theta^2$ and any $(\mathbf{z}^l,\bar{\mathbf{z}}^{l-1})\in \mathsf{B}^l_{\vartheta,C'}(\phi)$ we have
$$
\bar{\mathbb{E}}_{\phi,l}\left[
\frac{1}{N}\sum_{i\in\mathsf{S}_{k-1}^l}\|X_k^{i,l}-\bar{X}_k^{i,l-1}\|^r
\right]^{1/r} \leq C\left(\Delta_l^{\vartheta}+\|\theta-\bar{\theta}\|\right)
$$
and when $k\in\{1,\dots,T\}$
$$
\bar{\mathbb{E}}_{\phi,l}\left[
\frac{1}{N}\sum_{i\in\mathsf{S}_{k-1}^l}\|X_{k-1}^{a_{k-1}^{i,l},l}-\bar{X}_{k-1}^{\bar{a}_{k-1}^{i,l},l-1}\|^r
\right]^{1/r} \leq C\left(\Delta_l^{\vartheta}+\|\theta-\bar{\theta}\|\right).
$$
\end{lem}

\begin{proof}
The proof is very similar to that of \cite[Lemma A.2.]{ml_cont} and \cite[Lemma 13]{ub_grad},  except one uses
(A\ref{ass:coup_bridge}-\ref{ass:weight}) when is needed.  As such a proof is essentially repeating the calculations in the afore-mentioned papers,  it is omitted.
\end{proof}

\begin{lem}\label{lem:lem2}
Assume (A\ref{ass:coup_bridge}-\ref{ass:weight}, \ref{ass:weight_disc}-(i)).  Then  for any $(k,r,C')\in\{1,\dots,T\}\times[1,\infty)\times\mathbb{R}^+$ there exists a $C<\infty$ such that for any $(l,\vartheta,N,\phi)\in\mathbb{N}\times\mathbb{R}^+\times\{2,3,\dots\}\times\Theta^2$ and any $(\mathbf{z}^l,\bar{\mathbf{z}}^{l-1})\in \mathsf{B}^l_{\vartheta,C'}(\phi)$ we have
$$
\bar{\mathbb{E}}_{\phi,l}\Bigg[
\frac{1}{N}\sum_{i\in\mathsf{S}_{k-1}^l}
|
g_{\theta}(y_k|X_k^{i,l})R_{\theta,k-1,k}^l(C_{\theta,k-1,k}^l(X_{k-1}^{a_{k-1}^{i,l}},\mathbf{W}_{[k-1,k]}^{i,l},X_k^{i,l}))-
g_{\bar{\theta}}(y_k|\bar{X}_{k}^{i,l-1})
\times
$$
$$
R_{\bar{\theta},k-1,k}^{l-1}(C_{\bar{\theta},k-1,k}^{l-1}(\bar{X}_{k-1}^{\bar{a}_{k-1}^{i,l-1},l-1},\mathbf{W}_{[k-1,k]}^{i,l-1},\bar{X}_k^{i,l-1}))
|^r
\Bigg]^{1/r} \leq C\left(\Delta_l^{\tfrac{1}{2}\wedge\vartheta}+\|\theta-\bar{\theta}\|\right).
$$
\end{lem}

\begin{proof}
If $i\neq N$ then,  conditioning on the $\sigma-$field associated to the algorithm up-to and including the last resampling step,  one can apply (A\ref{ass:weight_disc}-(i)) followed by the $C_r-$inequality and Lemma \ref{lem:lem1}.  Otherwise one can use the fact that $(\mathbf{z}^l,\bar{\mathbf{z}}^{l-1})\in \mathsf{B}^l_{\vartheta,C'}(\phi)$.  This concludes the proof.
\end{proof}

\begin{lem}\label{lem:lem3}
Assume (A\ref{ass:coup_bridge}-\ref{ass:weight},  \ref{ass:weight_disc}-(i)).  Then  for any $(k,r,C')\in\{1,\dots,T\}\times[1,\infty)\times\mathbb{R}^+$ there exists a $C<\infty$ such that for any $(l,\vartheta,N,\phi)\in\mathbb{N}\times\mathbb{R}^+\times\{2,3,\dots\}\times\Theta^2$ and any 
$(\mathbf{z}^l,\bar{\mathbf{z}}^{l-1})\in \mathsf{B}^l_{\vartheta,C'}(\phi)$ we have
$$
1- \bar{\mathbb{E}}_{\phi,l}\left[\frac{\textrm{\emph{Card}}(\mathsf{S}_{k-1}^l)}{N}\right] \leq C
\left(\Delta_l^{\tfrac{1}{2}\wedge\vartheta}+\|\theta-\bar{\theta}\|\right).
$$
\end{lem}

\begin{proof}
The proof is similar to that of \cite[Lemma A.3.]{ml_cont} and therefore is not given.
\end{proof}

\begin{lem}\label{lem:lem4}
Assume (A\ref{ass:coup_bridge}-\ref{ass:weight_disc}).  Then  for any $(k,r,C')\in\{1,\dots,T\}\times[1,\infty)\times\mathbb{R}^+$ there exists a $C<\infty$ such that for any $(l,\vartheta,N,\phi)\in\mathbb{N}\times\mathbb{R}^+\times\{2,3,\dots\}\times\Theta^2$ and any 
$(\mathbf{z}^l,\bar{\mathbf{z}}^{l-1})\in \mathsf{B}^l_{\vartheta,C'}(\phi)\cap\mathsf{H}^l_{\vartheta,C'}(\phi)$ 
we have
$$
\bar{\mathbb{E}}_{\phi,l}\left[
\frac{1}{N}\sum_{i\in\mathsf{S}_{k-1}^l}\|
H_{l,k}(\theta,\mathbf{Z}_k^{i,l}) - 
H_{l-1,k}(\bar{\theta},\bar{\mathbf{Z}}_k^{i,l-1})
\|^r
\right]^{1/r}
\leq C
\left(\Delta_l^{\tfrac{1}{2}\wedge\vartheta}+\|\theta-\bar{\theta}\|\right).
$$
\end{lem}

\begin{proof}
Follows along the same lines as the proof of \cite[Lemma 16]{ub_grad},  except one must use (A\ref{ass:coup_bridge}-\ref{ass:weight_disc}) or Lemma \ref{lem:lem1} to complete the proof.  As the proof of this lemma is essentially the same up-to notational changes,  it is omitted.
\end{proof}

\begin{lem}\label{lem:lem5}
Assume (A\ref{ass:coup_bridge}(i)-(ii),  \ref{ass:weight}).   Then  for any $(k,r)\in\{0,\dots,T\}\times[1,\infty)$
there exists a $C<\infty$ such for any $(l,N,\phi,i)\in\mathbb{N}\times\mathbb{R}^+\times\{2,3,\dots\}\times\Theta^2\times\{1,\dots,N\}$ and any 
$(\mathbf{z}^l,\bar{\mathbf{z}}^{l-1})\in\mathsf{Z}^l\times\mathsf{Z}^{l-1}$  
we have
\begin{align*}
\bar{\mathbb{E}}_{\phi,l}\left[\|X_k^{i,l}\|^r\right] & \leq C\left(1+\sum_{t=0}^k\|x_k^{N,l}\|^r\right)\\
\bar{\mathbb{E}}_{\phi,l}\left[\|\bar{X}_k^{i,l-1}\|^r\right] & \leq C\left(1+\sum_{t=0}^k\|\bar{x}_k^{N,l-1}\|^r\right).
\end{align*}
\end{lem}

\begin{proof}
Follows using a simple inductive proof along with (A\ref{ass:coup_bridge}(i)-(ii),  \ref{ass:weight}).  See for instance \cite[Corollary 18]{ub_grad}.
\end{proof}

\begin{lem}\label{lem:lem6}
Assume (A\ref{ass:coup_bridge}-\ref{ass:grad_bound}).  Then  for any $(k,r,C')\in\{1,\dots,T\}\times[1,\infty)\times\mathbb{R}^+$ there exists a $C<\infty$ such that for any $(l,\vartheta,N,\phi,i)\in\mathbb{N}\times\mathbb{R}^+\times\{2,3,\dots\}\times\Theta^2\times\{1,\dots,N\}$ and any 
$(\mathbf{z}^l,\bar{\mathbf{z}}^{l-1})\in \mathsf{B}^l_{\vartheta,C'}(\phi)\cap\mathsf{H}^l_{\vartheta,C'}(\phi)$ 
we have
$$
\bar{\mathbb{E}}_{\phi,l}\left[
\|
H_{l,k}(\theta,\mathbf{Z}_k^{i,l}) - 
H_{l-1,k}(\bar{\theta},\bar{\mathbf{Z}}_k^{i,l-1})
\|^r
\right]^{1/r}
\leq C
\left(\Delta_l^{\tfrac{1}{2}\wedge\vartheta}+\|\theta-\bar{\theta}\|\right)^{1/r}.
$$
\end{lem}

\begin{proof}
Follows by considering $i\in\mathsf{S}_{k-1}^l$ or $i\in(\mathsf{S}_{k-1}^l)^c$.  In the former case one can use
Lemma \ref{lem:lem4} and in the latter (A\ref{ass:grad_bound}) followed by Lemma \ref{lem:lem3}. 
\end{proof}

\begin{rem}\label{rem:rem1}
If we assume that (A\ref{ass:coup_bridge}-\ref{ass:weight},  \ref{ass:weight_disc}-(i)) holds,  then one can prove the following results in a similar manner to Lemma \ref{lem:lem6} (see also \cite[Corollary 21]{ub_grad}).
\begin{enumerate}
\item{For any $(k,r,C')\in\{0,\dots,T\}\times[1,\infty)\times\mathbb{R}^+$ there exists a $C<\infty$ such that for any $(l,\vartheta,N,\phi,i,\delta)\in\mathbb{N}\times\mathbb{R}^+\times\{2,3,\dots\}\times\Theta^2\times\{1,\dots,N\}\times\mathbb{R}^+$ and any $(\mathbf{z}^l,\bar{\mathbf{z}}^{l-1})\in \mathsf{B}^l_{\vartheta,C'}(\phi)$ we have
$$
\bar{\mathbb{E}}_{\phi,l}\left[\|X_k^{i,l}-\bar{X}_k^{i,l-1}\|^r
\right]^{1/r} \leq C\left(\Delta_l^{\tfrac{1}{2}\wedge\vartheta}+\|\theta-\bar{\theta}\|\right)^{1/(r\{1+\delta\})}
\left(1
+\sum_{t=0}^k\|x_k^{N,l}\|^r
+\sum_{t=0}^k\|\bar{x}_k^{N,l-1}\|^r\right).
$$
}
\item{For any $(k,r,C')\in\{1,\dots,T\}\times[1,\infty)\times\mathbb{R}^+$ there exists a $C<\infty$ such that for any $(l,\vartheta,N,\phi,i)\in\mathbb{N}\times\mathbb{R}^+\times\{2,3,\dots\}\times\Theta^2\times\{1,\dots,N\}$ and any 
$(\mathbf{z}^l,\bar{\mathbf{z}}^{l-1})\in \mathsf{B}^l_{\vartheta,C'}(\phi)$ 
we have
$$
\bar{\mathbb{E}}_{\phi,l}\Bigg[
|
g_{\theta}(y_k|X_k^{i,l})R_{\theta,k-1,k}^l(C_{\theta,k-1,k}^l(X_{k-1}^{a_{k-1}^{i,l},l},\mathbf{W}_{[k-1,k]}^{i,l},X_k^{i,l}))-
g_{\bar{\theta}}(y_k|\bar{X}_{k}^{i,l-1})
\times
$$
$$
R_{\bar{\theta},k-1,k}^{l-1}(C_{\bar{\theta},k-1,k}^{l-1}(\bar{X}_{k-1}^{\bar{a}_{k-1}^{i,l-1},l-1},\mathbf{W}_{[k-1,k]}^{i,l-1},\bar{X}_k^{i,l-1}))
|^r
\Bigg]^{1/r}
\leq C
\left(\Delta_l^{\tfrac{1}{2}\wedge\vartheta}+\|\theta-\bar{\theta}\|\right)^{1/r}.
$$
}
\end{enumerate}
\end{rem}

\subsection{Results for Algorithm \ref{alg:dtccpbs}}\label{app:2}

We now consider Algorithm \ref{alg:dtccpbs},  where Step 4.~and Step 5.~have been replaced by \textbf{(FC)}.
Expectations associated to the probability law of the kernel are written $\check{\mathbb{E}}_{\phi,l}[\cdot]$.

\begin{lem}\label{lem:lem7}
Assume (A\ref{ass:coup_bridge}-\ref{ass:grad_bound}). 
\begin{enumerate}
\item{For any $(k,r,C')\in\{1,\dots,T\}\times[1,\infty)\times\mathbb{R}^+$ there exists a $C<\infty$ such that for any $(l,\vartheta,N,\phi)\in\mathbb{N}\times\mathbb{R}^+\times\{2,3,\dots\}\times\Theta^2$ and any 
$(\mathbf{z}^l,\bar{\mathbf{z}}^{l-1})\in \mathsf{B}^l_{\vartheta,C'}(\phi)\cap\mathsf{H}^l_{\vartheta,C'}(\phi)$ 
we have
$$
\check{\mathbb{E}}_{\phi,l}\left[
\|
H_{l,k}(\theta,\mathbf{Z}_k^{j_T^l,l}) - 
H_{l-1,k}(\bar{\theta},\bar{\mathbf{Z}}_k^{\bar{j}_T^{l-1},l-1})
\|^r
\right]^{1/r}
\leq C
\left(\Delta_l^{\tfrac{1}{2}\wedge\vartheta}+\|\theta-\bar{\theta}\|\right)^{1/r}.
$$}
\item{For any $(k,r,C')\in\{0,\dots,T\}\times[1,\infty)\times\mathbb{R}^+$ there exists a $C<\infty$ such that for any $(l,\vartheta,N,\phi,\delta)\in\mathbb{N}\times\mathbb{R}^+\times\{2,3,\dots\}\times\Theta^2\times\mathbb{R}^+$ and any $(\mathbf{z}^l,\bar{\mathbf{z}}^{l-1})\in \mathsf{B}^l_{\vartheta,C'}(\phi)$ we have
$$
\check{\mathbb{E}}_{\phi,l}\left[\|X_k^{j_k^l,l}-\bar{X}_k^{\bar{j}_k^{l-1},l-1}\|^r
\right]^{1/r} \leq C\left(\Delta_l^{\tfrac{1}{2}\wedge\vartheta}+\|\theta-\bar{\theta}\|\right)^{1/(r\{1+\delta\})}
\left(1
+\sum_{t=0}^k\|x_t^{N,l}\|^r
+\sum_{t=0}^k\|\bar{x}_t^{N,l-1}\|^r\right).
$$
}
\item{For any $(k,r,C')\in\{1,\dots,T\}\times[1,\infty)\times\mathbb{R}^+$ there exists a $C<\infty$ such that for any $(l,\vartheta,N,\phi)\in\mathbb{N}\times\mathbb{R}^+\times\{2,3,\dots\}\times\Theta^2$ and any 
$(\mathbf{z}^l,\bar{\mathbf{z}}^{l-1})\in \mathsf{B}^l_{\vartheta,C'}(\phi)$ 
we have
$$
\check{\mathbb{E}}_{\phi,l}\Bigg[
|
g_{\theta}(y_k|X_k^{j_k^l,l})R_{\theta,k-1,k}^l(C_{\theta,k-1,k}^l(X_{k-1}^{j_{k-1}^l,l},\mathbf{W}_{[k-1,k]}^{j_k^l,l},X_k^{j_k^l,l}))-
g_{\bar{\theta}}(y_k|\bar{X}_{k}^{\bar{j}_{k}^{l-1},l-1})
\times
$$
$$
R_{\bar{\theta},k-1,k}^{l-1}(C_{\bar{\theta},k-1,k}^{l-1}(\bar{X}_{k-1}^{\bar{j}_{k-1}^{l-1},l-1},\mathbf{W}_{[k-1,k]}^{\bar{j}_{k}^{l-1},l-1},\bar{X}_k^{\bar{j}_k^{l-1},l-1}))
|^r
\Bigg]^{1/r}
\leq C
\left(\Delta_l^{\tfrac{1}{2}\wedge\vartheta}+\|\theta-\bar{\theta}\|\right)^{1/r}.
$$}
\end{enumerate}
\end{lem}

\begin{proof}
Can be established using Lemma \ref{lem:lem6},  Remark \ref{rem:rem1} and the arguments used for \cite[Corollary24, 27]{ub_grad}. 
\end{proof}

\begin{lem}\label{lem:lem8}
Assume (A\ref{ass:coup_bridge}-\ref{ass:grad_bound}).  Then for any $(T,C')\in\mathbb{N}\times\mathbb{R}^+$ there exists a $C<\infty$ such that for any $(l,\vartheta,N,\phi,\delta,\zeta)\in\mathbb{N}\times\mathbb{R}^+\times\{2,3,\dots\}\times\Theta^2\times\mathbb{R}^+\times(0,\frac{\tfrac{1}{2}\wedge\vartheta}{\vartheta(1+\delta)})$ and any 
$(\mathbf{z}^l,\bar{\mathbf{z}}^{l-1})\in \mathsf{B}^l_{\vartheta,C'}(\phi)\cap\mathsf{H}^l_{\vartheta,C'}(\phi)$ 
we have
$$
\check{\mathbb{E}}_{\phi,l}\left[\mathbb{I}_{(\mathsf{B}^l_{\vartheta,C'}(\phi)\cap\mathsf{H}^l_{\vartheta,C'}(\phi))^c}(\mathbf{Z}_T^{j_T^l,l},\bar{\mathbf{Z}}_T^{\bar{j}_T^{l-1},l-1})\right] \leq 
C\left(\Delta_l^{\tfrac{1}{2}\wedge\vartheta}+\|\theta-\bar{\theta}\|\right)^{1/(1+\delta)-\zeta\vartheta}
\left(1
+\sum_{k=0}^T\|x_k^{N,l}\|^{\vartheta}
+\sum_{k=0}^T\|\bar{x}_k^{N,l-1}\|^{\vartheta}\right).
$$
\end{lem}

\begin{proof}
Can be established using Markov's inequality and Lemma \ref{lem:lem7}; see the proof of \cite[Lemma 28]{ub_grad}.
\end{proof}

\subsection{Results Associated to the Initialization}\label{app:3}

This section gives some results associated to the probability measure \eqref{eq:p_initial_coup},  
$\check{\mathbb{P}}^l_{\theta}$ and we write expectations associated to $\check{\mathbb{P}}^l_{\theta}$ as
$\check{\mathbb{E}}^l_{\theta}$.

\begin{lem}\label{lem:lem9}
Assume (A\ref{ass:coup_bridge},\ref{ass:weight_disc}).
\begin{enumerate}
\item{For any $(k,r)\in\{1,\dots,T\}\times[1,\infty)$ there exists a $C<\infty$ such that for any $(l,\theta)\in\mathbb{N}\times\Theta$ 
we have
$$
\check{\mathbb{E}}_{\phi}^l\left[
\|
H_{l,k}(\theta,\mathbf{Z}_k^{l}) - 
H_{l-1,k}(\theta,\bar{\mathbf{Z}}_k^{l-1})
\|^r
\right]^{1/r}
\leq C
\Delta_l^{\tfrac{1}{2}}.
$$
}
\item{For any $(k,r)\in\{1,\dots,T\}\times[1,\infty)$ there exists a $C<\infty$ such that for any $(l,\theta)\in\mathbb{N}\times\Theta$ 
we have
$$
\check{\mathbb{E}}_{\theta}^l\Bigg[
|
g_{\theta}(y_k|X_k^{l})R_{\theta,k-1,k}^l(C_{\theta,k-1,k}^l(X_{k-1}^{l},\mathbf{W}_{[k-1,k]}^{l},X_k^{j_k^l,l}))-
g_{\bar{\theta}}(y_k|\bar{X}_{k}^{l-1})
\times
$$
$$
R_{\theta,k-1,k}^{l-1}(C_{\theta,k-1,k}^{l-1}(\bar{X}_{k-1}^{l-1},\mathbf{W}_{[k-1,k]}^{l-1},\bar{X}_k^{l-1}))
|^r
\Bigg]^{1/r}
\leq C
\Delta_l^{\tfrac{1}{2}}.
$$}
\end{enumerate}
\end{lem}

\begin{proof}
Follows easily by recursive application of the relevant assumptions and noting that 
$\check{\mathbb{E}}_{\theta}^l[\|X_k^l-\bar{X}_k^{l-1}\|^r]=0$ by using (A\ref{ass:coup_bridge}).
\end{proof}

\begin{lem}\label{lem:lem10}
Assume (A\ref{ass:coup_bridge},\ref{ass:weight_disc}). Then for any $(T,C')\in\mathbb{N}\times\mathbb{R}^+$ there exists a $C<\infty$ such that for any $(l,\vartheta,\theta)\in\mathbb{N}\times(0,\tfrac{1}{2})\times\Theta$ 
we have
$$
\check{\mathbb{E}}_{\theta}^l\left[\mathbb{I}_{(\mathsf{B}^l_{\vartheta,C'}(\theta,\theta)\cap\mathsf{H}^l_{\vartheta,C'}(\theta,\theta))^c}(\mathbf{Z}^{l},\bar{\mathbf{Z}}^{l-1})\right] \leq 
C\Delta_l^{\tfrac{1}{2}-\vartheta}.
$$
\end{lem}

\begin{proof}
Can be established by using Markov's inequality and Lemma \ref{lem:lem9}. 
\end{proof}

\subsection{Results Associated to Algorithm \ref{alg:USMA}}\label{app:4}

We now consider Algorithm \ref{alg:USMA} and write expectations associated to the law as $\mathbb{E}[\cdot]$.
In many of the results below,  $l\in\mathbb{N}$ is given and the afore-mentioned expectation is conditioning on $l$ having been sampled and $n\in\mathbb{N}$ the time index being allowed to run indefinitely.  Although in practice we stop the simulation at time $N_p$ this is ignored,  simply because it does not make any real impact in our final proof.
We introduce the sets for $(l,\vartheta,C)\in\mathbb{N}\times\mathbb{R}^+\times\mathbb{R}^+$
\begin{eqnarray*}
\overline{\mathsf{B}}_{\vartheta,C}^l & := & \{(\mathbf{z}^l,\bar{\mathbf{z}}^{l-1},\theta,\bar{\theta})\in\mathsf{Z}^l`\times
\mathsf{Z}^{l-1}\times\Theta^2: \|\theta-\bar{\theta}\|\leq C\Delta_l^{\vartheta},(\mathbf{z}^l,\bar{\mathbf{z}}^{l-1})\in
\mathsf{B}_{\vartheta,C}^l(\theta,\bar{\theta})
\}\\
\overline{\mathsf{H}}_{\vartheta,C}^l & := & \{(\mathbf{z}^l,\bar{\mathbf{z}}^{l-1},\theta,\bar{\theta})\in\mathsf{Z}^l`\times
\mathsf{Z}^{l-1}\times\Theta^2: \|\theta-\bar{\theta}\|\leq C\Delta_l^{\vartheta},(\mathbf{z}^l,\bar{\mathbf{z}}^{l-1})\in
\mathsf{H}_{\vartheta,C}^l(\theta,\bar{\theta})
\}
\end{eqnarray*}

\begin{lem}\label{lem:lem11}
Assume (A\ref{ass:coup_bridge}-\ref{ass:grad_bound}).  Then for any $(T,r,C')\in\mathbb{N}\times[1,\infty)\times\mathbb{R}^+$ there exists a $C<\infty$ such that for any $(l,\vartheta,N,n)\in\mathbb{N}\times\mathbb{R}^+\times\{2,3,\dots\}\times\mathbb{N}$ we have
$$
\mathbb{E}\left[
\|
H_{l}(\theta_{n-1}^l,\mathbf{Z}^{n,l}) - 
H_{l-1}(\theta_{n-1}^{l-1},\bar{\mathbf{Z}}^{n,l-1})
\|^r\mathbb{I}_{\overline{\mathsf{B}}_{\vartheta,C'}^l\cap \overline{\mathsf{H}}_{\vartheta,C'}^l}(\mathbf{Z}^{n-1,l},\mathbf{Z}^{n-1,l},\theta_{n-1}^l,\theta_{n-1}^{l-1})
\right]^{1/r}
\leq C\left(\Delta_l^{\tfrac{1}{2}\wedge\vartheta}\right)^{1/r}.
$$
\end{lem}

\begin{proof}
Follows from Lemma \ref{lem:lem7} 1..
\end{proof}

\begin{lem}\label{lem:lem12}
Assume (A\ref{ass:coup_bridge}-\ref{ass:grad_bound}).  Then for any $(T,r,C')\in\mathbb{N}\times[1,\infty)\times\mathbb{R}^+$ there exists a $C<\infty$ such that for any $(l,\vartheta,N,n)\in\mathbb{N}\times\mathbb{R}^+\times\{2,3,\dots\}\times\mathbb{N}$ we have
$$
\mathbb{E}\left[
\|\theta_{n}^l - \theta_{n}^{l-1}\|^r\mathbb{I}_{\overline{\mathsf{B}}_{\vartheta,C'}^l\cap \overline{\mathsf{H}}_{\vartheta,C'}^l}(\mathbf{Z}^{n-1,l},\mathbf{Z}^{n-1,l},\theta_{n-1}^l,\theta_{n-1}^{l-1})
\right]^{1/r}
\leq C\left(\Delta_l^{\vartheta}+\left(\Delta_l^{\tfrac{1}{2}\wedge\vartheta}\right)^{1/r}\right).
$$
\end{lem}

\begin{proof}
Follows easily from the recursive update of $\theta_n^l,\theta_n^{l-1}$ and Lemma \ref{lem:lem11}.
\end{proof}

\begin{lem}\label{lem:lem13}
Assume (A\ref{ass:coup_bridge}-\ref{ass:mom_cond}).  Then for any $(T,C')\in\mathbb{N}\times\mathbb{R}^+$ there exists a $C<\infty$ such that for any $(l,\vartheta,N,\delta,\zeta,n)\in\mathbb{N}\times(0,1)\times\{2,3,\dots\}\times\mathbb{R}^+\times(0,\frac{\tfrac{1}{2}\wedge\vartheta}{\vartheta(1+\delta)})\times\mathbb{N}$ 
we have
$$
\mathbb{E}\left[\mathbb{I}_{(\overline{\mathsf{B}}^l_{\vartheta,C'}\cap\overline{\mathsf{H}}^l_{\vartheta,C'})^c}(\mathbf{Z}^{n,l},\bar{\mathbf{Z}}^{n,l-1},\theta_n^l,
\theta_{n}^{l-1})
\mathbb{I}_{\overline{\mathsf{B}}^l_{\vartheta,C'}\cap\overline{\mathsf{H}}^l_{\vartheta,C'}}(\mathbf{Z}^{n-1,l},\bar{\mathbf{Z}}^{n-1,l-1},\theta_{n-1}^l,
\theta_{n-1}^{l-1})
\right] \leq 
C\left(\Delta_l^{\tfrac{1}{2}\wedge\vartheta}\right)^{\bar{\vartheta}}
$$
where $\bar{\vartheta}=\min\{1/(1+\delta)-\zeta\vartheta,\zeta(1-\vartheta)\}$.
\end{lem}

\begin{proof}
Consider
$$
\mathbb{E}[
\mathbb{I}_{ \mathsf{C}^l_{\vartheta,C'}}(\theta_n^l,\theta_{n}^{l-1})
\mathbb{I}_{\overline{\mathsf{B}}^l_{\vartheta,C'}\cap\overline{\mathsf{H}}^l_{\vartheta,C'}}(\mathbf{Z}^{n-1,l},\bar{\mathbf{Z}}^{n-1,l-1},\theta_{n-1}^l,
\theta_{n-1}^{l-1})
]
$$
where $\mathsf{C}^l_{\vartheta,C'}=\{(\theta,\bar{\theta})\in\Theta^2:\|\theta-\bar{\theta}\|>C'\Delta_l^{\vartheta}\}$.  Then using Markov's inequality along with Lemma \ref{lem:lem12} we have that
$$\mathbb{E}[
\mathbb{I}_{ \mathsf{C}^l_{\vartheta,C'}}(\theta_n^l,\theta_{n}^{l-1})
\mathbb{I}_{\overline{\mathsf{B}}^l_{\vartheta,C'}\cap\overline{\mathsf{H}}^l_{\vartheta,C'}}(\mathbf{Z}^{n-1,l},\bar{\mathbf{Z}}^{n-1,l-1},\theta_{n-1}^l,
\theta_{n-1}^{l-1})
]
\leq C\left(\Delta_l^{\tfrac{1}{2}\wedge\vartheta}\right)^{\zeta(1-\vartheta)}.
$$
If one considers for instance
$$
\mathbb{E}[
\mathbb{I}_{ \mathsf{D}^l_{\vartheta,C',k}(\theta_n^l,\theta_{n}^{l-1})}(\mathbf{Z}^{n,l},\bar{\mathbf{Z}}^{n,l-1})
\mathbb{I}_{\overline{\mathsf{B}}^l_{\vartheta,C'}\cap\overline{\mathsf{H}}^l_{\vartheta,C'}}
(\mathbf{Z}^{n-1,l},\bar{\mathbf{Z}}^{n-1,l-1},\theta_{n-1}^l,
\theta_{n-1}^{l-1})
]
$$
where 
$$
\mathsf{D}^l_{\vartheta,C',k}(\theta_n^l,\theta_{n}^{l-1})
=\{\mathbf{z}\times\bar{\mathbf{z}}\in\mathsf{Z}^l\times\mathsf{Z}^{l-1}:\|H_{l,k}^l(\theta_n^l,\mathbf{z}_k) - H_{l-1,k}^l(\theta_{n}^{l-1},\bar{\mathbf{z}}_k)\| > C'(\Delta_l^{\vartheta}+
\|\theta_n^l-\theta_{n}^{l-1}\|
)
\}.
$$
Then clearly
$$
\mathsf{D}^l_{\vartheta,C',k}(\theta_n^l,\theta_{n}^{l-1})
\subseteq
\{\mathbf{z}\times\bar{\mathbf{z}}\in\mathsf{Z}^l\times\mathsf{Z}^{l-1}:\|H_{l,k}^l(\theta_n^l,\mathbf{z}_k) - H_{l-1,k}^l(\theta_{n}^{l-1},\bar{\mathbf{z}}_k)\| > C'\Delta_l^{\vartheta}
\}.
$$
Then, just as for the proof of Lemma \ref{lem:lem8} one can use a Markov inequality argument,  combining with (A\ref{ass:mom_cond}) when needed.
\end{proof}

\begin{lem}\label{lem:lem14}
Assume (A\ref{ass:coup_bridge}-\ref{ass:mom_cond}).  Then for any $(T,C')\in\mathbb{N}\times\mathbb{R}^+$ there exists a $C<\infty$ such that for any $(l,\vartheta,N,\delta,\zeta,n)\in\mathbb{N}\times(0,\tfrac{1}{2})\times\{2,3,\dots\}\times\mathbb{R}^+\times(0,\frac{1}{(1+\delta)})\times\mathbb{N}$ 
we have
$$
\mathbb{E}\left[\mathbb{I}_{(\overline{\mathsf{B}}^l_{\vartheta,C'}\cap\overline{\mathsf{H}}^l_{\vartheta,C'})^c}(\mathbf{Z}^{n,l},\bar{\mathbf{Z}}^{n,l-1},\theta_n^l,
\theta_{n}^{l-1})\right] \leq 
C(n+1)\Delta_l^{\tilde{\vartheta}}
$$
where 
\begin{equation}\label{eq:exp_def}
\tilde{\vartheta} = \min\{
\min\{1/(1+\delta)-\zeta\vartheta,\zeta(1-\vartheta)\}
\vartheta,
\tfrac{1}{2}-\vartheta
\}.
\end{equation}
\end{lem}

\begin{proof}
The proof is by induction,  where the initialization follows by Lemma \ref{lem:lem10}. The induction follows
by using Lemma \ref{lem:lem13} along with the induction hypothesis.  See for instance the proof of \cite[Lemma 35]{ub_grad}.
\end{proof}

\begin{lem}\label{lem:lem15}
Assume (A\ref{ass:coup_bridge}-\ref{ass:mom_cond}).  Then for any $(T,r,C')\in\mathbb{N}\times[1,\infty)\times\mathbb{R}^+$ there exists a $C<\infty$ such that for any $(l,\vartheta,N,\delta,\zeta,n)\in\mathbb{N}\times(0,\tfrac{1}{2})\times\{2,3,\dots\}\times\mathbb{R}^+\times(0,\frac{1}{(1+\delta)})\times\mathbb{N}$ 
we have
$$
\mathbb{E}\left[
\|
H_{l}(\theta_{n-1}^l,\mathbf{Z}^{n,l}) - 
H_{l-1}(\theta_{n-1}^{l-1},\bar{\mathbf{Z}}^{n,l-1})
\|^r
\right]^{1/r}
\leq C(n+1)\Delta_l^{\tilde{\vartheta}_r}
$$
where $\tilde{\vartheta}_r=\min\{\vartheta/r,\tilde{\vartheta}\}$ and $\tilde{\vartheta}$ is as \eqref{eq:exp_def}.
\end{lem}

\begin{proof}
Follows by using Lemmata \ref{lem:lem11},  \ref{lem:lem14}.
\end{proof}

\subsection{Proof of Theorem \ref{theo:main}}\label{app:main}

\begin{proof}
Set $\overline{\theta}_{\star}(p,l)=\mathbb{P}_{\mathtt{P}}(p)\mathbb{P}_{\mathtt{L}}(l)\widehat{\theta}_{\star}$.  Then to prove the result it is enough to show that
$$
\sum_{(l,p)\in\mathbb{N}_0^2}\frac{1}{\mathbb{P}_{\mathtt{P}}(p)\mathbb{P}_{\mathtt{L}}(l)}
\mathbb{E}[\|\overline{\theta}_{\star}(p,l)\|^2] <\infty
$$
see for instance \cite[Theorem 3]{matti}. We first consider the case $(p,l)\in\mathbb{N}$ in the summation,  in which case one
has that 
$$
\overline{\theta}_{\star}(p,l) = \sum_{n=N_{p-1}-1}^{N_p-1}\gamma_n\left\{
H_{l}(\theta_{n-1}^l,\mathbf{z}^{n,l}) - 
H_{l-1}(\theta_{n-1}^{l-1},\bar{\mathbf{z}}^{n,l-1})
\right\}.
$$
Therefore on applying Lemma \ref{lem:lem15} along with the Minkowski inequality we have that
$$
\sum_{(l,p)\in\mathbb{N}^2}\frac{1}{\mathbb{P}_{\mathtt{P}}(p)\mathbb{P}_{\mathtt{L}}(l)}
\mathbb{E}[\|\overline{\theta}_{\star}(p,l)\|^2]
\leq
C\sum_{(l,p)\in\mathbb{N}^2}\frac{\Delta_l^{2\tilde{\vartheta}_2}}{\mathbb{P}_{\mathtt{P}}(p)\mathbb{P}_{\mathtt{L}}(l)} 
\left\{\sum_{n=N_{p-1}-1}^{N_p-1}\gamma_n(n+1)\right\}^2.
$$
Next when $l\in\mathbb{N}$ and $p=0$ we have
$$
\overline{\theta}_{\star}(p,l) = \sum_{n=1}^{N_0-1}\gamma_n\left\{
H_{l}(\theta_{n-1}^l,\mathbf{z}^{n,l}) - 
H_{l-1}(\theta_{n-1}^{l-1},\bar{\mathbf{z}}^{n,l-1}).
\right\}
$$
Again,  using Lemma \ref{lem:lem15} along with Minkowski,  it follows that
$$
\sum_{l\in\mathbb{N}}\frac{1}{\mathbb{P}_{\mathtt{P}}(0)\mathbb{P}_{\mathtt{L}}(l)}
\mathbb{E}[\|\overline{\theta}_{\star}(0,l)\|^2]
\leq
C\sum_{l\in\mathbb{N}}\frac{\Delta_l^{2\tilde{\vartheta}_2}}{\mathbb{P}_{\mathtt{P}}(0)\mathbb{P}_{\mathtt{L}}(l)} 
\left\{\sum_{n=1}^{N_0-1}\gamma_n(n+1)\right\}^2.
$$
Therefore we have proved that
$$
\sum_{(l,p)\in\mathbb{N}_0^2}\frac{1}{\mathbb{P}_{\mathtt{P}}(p)\mathbb{P}_{\mathtt{L}}(l)}
\mathbb{E}[\|\overline{\theta}_{\star}(p,l)\|^2] \leq
$$
$$
C
\left\{
\frac{1}{\mathbb{P}_{\mathtt{P}}(0)\mathbb{P}_{\mathtt{L}}(0)}
\mathbb{E}[\|\overline{\theta}_{\star}(0,0)\|^2] +
\sum_{l\in\mathbb{N}}\frac{\Delta_l^{2\tilde{\vartheta}_2}}{\mathbb{P}_{\mathtt{P}}(0)\mathbb{P}_{\mathtt{L}}(l)} 
\left\{\sum_{n=1}^{N_0-1}\gamma_n(n+1)\right\}^2 +
\sum_{(l,p)\in\mathbb{N}^2}\frac{\Delta_l^{2\tilde{\vartheta}_2}}{\mathbb{P}_{\mathtt{P}}(p)\mathbb{P}_{\mathtt{L}}(l)} 
\left\{\sum_{n=N_{p-1}-1}^{N_p-1}\gamma_n(n+1)\right\}^2
\right\}.
$$
If one sets $N_p=p+1$,  $\gamma_n=(n+1)^{-(2+\kappa)}$,  $\kappa>0$, $\mathbb{P}_{\mathtt{P}}(p)\propto (p+1)^{-(1+\kappa)}$ and $\mathbb{P}_{\mathtt{L}}(l)\propto\Delta_l^{\tilde{\vartheta}_2}$ then one has the desired result.
\end{proof}

\end{document}